\documentclass[12pt]{amsart}
\usepackage{amssymb,amsmath,mathrsfs,enumerate,mathtools,comment}
\usepackage{tikz-cd}
\usepackage[pdftex,bookmarks=true]{hyperref}
\usepackage{adjustbox}
\usepackage{hhline}
\usepackage{enumerate,graphicx}
\usepackage{filecontents}
\newtheorem*{theorem*}{Theorem}
\theoremstyle{plain}

\newtheorem{theorem}{Theorem}[section]

\newtheorem{proposition}[theorem]{Proposition}
\newtheorem{lemma}[theorem]{Lemma}
\newtheorem{corollary}[theorem]{Corollary}
\theoremstyle{definition}

\newtheorem{example}[theorem]{Example}

\newtheorem{remark}[theorem]{Remark}
\numberwithin{equation}{section}

\newcommand{\abs}[1]{\lvert#1\rvert}
\newcommand{\norm}[1]{\lVert#1\rVert}

\newcommand{\bignorm}[1]{\bigl\lVert#1\bigr\rVert}

\newcommand{\Bignorm}[1]{\Bigl\lVert#1\Bigr\rVert}

\usepackage{array}

\usepackage{bbm}
\usepackage{float}

\newcommand{\N}{{\mathbb N}}
\newcommand{\E}{{\mathbb E}}
\newcommand{\bP}{{\mathbb P}}
\newcommand{\R}{{\mathbb R}}
\newcommand{\cX}{{\mathcal X}}
\newcommand{\cY}{{\mathcal Y}}

\newcommand{\one}{\mathbf{1}}
\newcommand{\ol}{\overline}

\title{Automatic Fatou Property of Law-invariant Risk Measures}

\author[S.~Chen]{Shengzhong Chen}
\address{Department of Mathematics, Ryerson University, 350 Victoria Street, Toronto, Canada M5B 2K3}
\email{sz.chen@ryerson.ca}
\author[N.~Gao]{Niushan Gao}
\address{Department of Mathematics, Ryerson University, 350 Victoria Street, Toronto, Canada M5B 2K3}
\email{niushan@ryerson.ca}

\author[D.~Leung]{Denny H.~Leung}
\address{Department of Mathematics, National University of Singapore, Singapore 117543}
\email{dennyhl@u.nus.edu}

\author[L.~Li]{Lei Li}
\address{School of Mathematical Sciences and LPMC,  Nankai University, Tianjin 300071, China
}
\email{leilee@nankai.edu.cn}

\thanks{The first author is supported by an NSERC fellowship. The second author acknowledges support of an NSERC Discovery Grant. The fourth author is supported by NSF of China (12171251).}



\date{\today}

\begin{document}
\maketitle
\begin{abstract}
In the paper we investigate automatic Fatou property of law-invariant risk measures on a rearrangement-invariant function space $\cX$ other than $L^\infty$. The main result is the following characterization:  Every real-valued, law-invariant, coherent risk measure on $\cX$ has the Fatou property at every random variable $X\in \cX$ whose negative tails have vanishing norm (i.e., $\lim_n\norm{X\one_{\{X\leq -n\}}}=0$) if and only if  $\cX$ satisfies the Almost Order Continuous Equidistributional Average (AOCEA) property, namely, $\mathrm{d}(\mathcal{CL}(X),\cX_a) =0$ for any $X\in \cX_+$, where $ \mathcal{CL}(X)$ is the convex hull of all random variables having the same distribution as $X$ and  $\cX_a=\{X\in\cX:\lim_n\norm{X\one_{\{\abs{X}\geq n\}}}=0\}$.
As a consequence, we   show that under the AOCEA property,   every real-valued, law-invariant, coherent risk measure on $\cX$ admits a tractable dual representation at every $X\in \cX$ whose negative tails have vanishing norm.
Furthermore, we show that  the AOCEA property is satisfied by most classical model spaces, including    Orlicz spaces, and therefore the foregoing results have wide applications.
\end{abstract}

\section{Introduction}

The axiomatic theory of risk measures has been an active research area ever since their introduction in the landmark paper by Artzner et al \cite{ArtznerDelbaenEberHeath1999}. Many properties of risk measures have been investigated in depth, and rigorous debates have taken place as to which of these properties ought to be regarded as a natural part of the definition of risk measures.  Among the properties considered, law invariance is regarded as highly relevant in practice as real-world computations of risk are often based on probability distributions of financial positions. As a matter of fact, most concrete risk measures (such as Value-at-risk, Expected shortfall,  Haezendonck-Goovaerts risk measures) and special classes of risk measures (such as distortion and quantile-based risk measures) are indeed law invariant. Research on general law-invariant risk measures has been intense and has produced many profound results; see, e.g., Bellini et al \cite{BC,BCb}, Chen et al \cite{CGX}, Filipovi\'c and  Svindland \cite{FS08,FS12}, Gao et al \cite{GaoLeungMunariXanthos2017}, Jouini et al \cite{Joui06,Joui08}, Kr\"{a}tschmer et al \cite{KSZ}, Kusuoka \cite{K01}, Liu et al \cite{LW}, Wang and Zitikis \cite{WR}, and Weber \cite{W}.

\medskip

Below we provide a brief review of one specific direction of research    that highlights the importance of law invariance  and motivates the present paper. 
In the context of a convex risk measure, or more generally, a  convex functional $\rho$ defined on a   function space $\cX$, lower semicontinuity with respect to a weak topology $\sigma(\cX,\cY)$ determined by a dual space $\cY$ leads to a dual representation of $\rho$, thanks to the well-known Fenchel-Moreau Duality:
\begin{align}\label{dualgeneral}
    \rho(X) = \sup_{Y\in \cY}\big(\langle X,Y\rangle -\rho^*(Y)\big), \quad X\in \cX,
\end{align}
where 
\[ \rho^*(Y) = \sup_{X\in \cX}\big(\langle X,Y\rangle - \rho(X)\big),\quad Y\in\cY.\]
In general, such dual representations play an important role in optimization
and portfolio selection.
For practical purposes, it is particularly desirable to be able to take the dual space $\cY$ as a space of functions (as opposed to abstract linear functionals). In this case, the representation \eqref{dualgeneral} is generally deemed as tractable or manageable.  
A concrete, and more verifiable, alternative to topological lower semicontinuity is order lower semicontinuity. It is usually termed as the Fatou property in the actuarial and risk management literature.  Specifically, a functional $\rho :\cX\to (-\infty,\infty]$ defined on a  function space $\cX$ is said to satisfy the {\em Fatou property}  if 
\[ \rho(X) \leq \liminf \rho(X_k) \text{ whenever } X_k \to  X \text{ a.s.\ and } |X_k| \leq Y \text{ for some $Y \in \cX$}.
\]
In the first result of its kind, Delbaen \cite{D} showed that  if a proper   convex functional $\rho$ on $L^\infty$ satisfies the (easier to verify) Fatou property, then it is $\sigma(L^\infty,L^1)$ lower semicontinuous and thus enables a tractable dual representation with $L^1$ as dual space by the Fenchel-Moreau Duality. 
Biagini and Frittelli \cite{BF} and Delbaen and Owari \cite{DO} then  drew attention to  the following natural question:\smallskip
\begin{center}
\em if a proper convex functional $\rho$  on a general function space $\cX$ satisfies\\ the Fatou property, must it admit a tractable dual representation? 
\end{center} \smallskip
Gao et al \cite{GLX2} showed that in general the answer    is no for the class of  Orlicz spaces.
However, it was later  proved that, surprisingly, if $\rho$ is additionally \emph{law invariant} then  the answer is yes! See Gao et al \cite{GaoLeungMunariXanthos2017} for Orlicz spaces and 
Tantrawan and Leung \cite{TL} for  general rearrangement-invariant  spaces. 
These results highlight the superior behavior that law invariance brings to the study of risk measures.
Additional  developments on the Fatou property and tractable dual representations of law-invariant risk measures can be found in the recent  works of Bellini et al \cite{BC, BCb, BCc},  Chen et al \cite{CGX},   Filipovi\'c et al \cite{FS12}, Gao et al \cite{GMX, GX}, Liu et al \cite{LCLW} and Svindland \cite{Svindland}. 

Perhaps the most striking result demonstrating the power of law invariance in connection with the Fatou property is   the following theorem.
\begin{theorem*}[Jouini et al \cite{Joui06}]
A real-valued, convex, decreasing, law-invariant functional on $L^\infty$ has the Fatou property. Consequently, it is $\sigma(L^\infty, L^1)$ lower semicontinuous and admits a dual representation via $L^1$.
\end{theorem*}
\noindent This theorem can be viewed as a result on automatic continuity since the Fatou property is a type of continuity property (order lower semicontinuity, to be precise). 
Automatic continuity has   long  been an interesting research topic and probably has its roots in the well-known fact that a real-valued convex function on an open interval is continuous. In infinite-dimensional spaces,  Birkhoff's Theorem states that a positive linear functional on a Banach lattice is norm continuous. It was later extended to the following celebrated theorem for real-valued convex functionals. 

\begin{theorem*}[Ruszczy\'{n}ski and Shapiro \cite{RS06}] A real-valued, convex, decreasing functional on a Banach lattice is norm continuous.
\end{theorem*}
\noindent This result has drawn extensive attention in optimization, operations research and risk management in the past decade. We remark  that Fatou property and norm continuity do not imply each other but Fatou property is stronger than norm lower semicontinuity.
We refer the reader to Biagini and Frittelli  \cite{BF}, Farkas et al \cite{FarkasKochMunari2014} and Munari \cite{Munari} for further results on automatic norm continuity properties of risk measures.

\medskip

The present paper aims at investigating automatic Fatou property  of law-invariant risk measures on general model spaces. In recent years, study of risk measures has been extended from Lebesgue spaces to more general settings such as Orlicz spaces (see, e.g., Biagini and Frittelli \cite{BF}, Cheridito and Li \cite{CheriditoLi2008,CheriditoLi2009}, Kr\"{a}tschmer et al \cite{KSZ}, Gao et al \cite{GaoLeungMunariXanthos2017}) and other more general spaces (see, e.g., Bellini \cite{BC,BCb}, Chen et al \cite{CGX}, Drapeau and Kupper \cite{DK}, Farkas et al \cite{FarkasKochMunari2014}, and Frittelli and Rosazza Gianin \cite{FR}). In this paper, we   adopt general rearrangement-invariant function spaces  as our model space (see the precise definition in the ``Notation and Facts" subsection below). The reason is two-fold. Firstly, among all rearrangement-invariant spaces, we characterize precisely  the spaces where the Fatou property automatically holds. This provides   genuine insights into exactly what makes the Fatou property automatic. Secondly, the proofs on general model spaces are not more complicated than those  on special model spaces such as Orlicz spaces. Quite on the contrary, working on special spaces often obscures the essential ingredients in the proofs with space-related technicalities.

The paper is structured as follows. Let $\cX$ be a rearrangement-invariant space   other than $L^\infty$ over a non-atomic probability space (we refer to the next subsection for the reasons for excluding $L^\infty$).
In Section 2, Example \ref{counterexam} shows that   there is always a real-valued law-invariant coherent risk measure on $\cX$ that fails the Fatou property at every $X\in \cX$ whose negative  tails do not have vanishing norm, i.e., $\lim_n\norm{X\one_{\{X\leq -n\}}}>0$.  Since   automatic Fatou property for a general law-invariant  coherent risk measure  cannot be expected over the whole space in general, we seek to determine the spaces where the automatic Fatou property has maximal validity.
The Main Theorem (Theorem \ref{main-thm}) characterizes the  spaces $\cX$ on which every real-valued, law-invariant, coherent risk measure is automatically  Fatou at all $X\in \cX$ whose negative tails have vanishing norm.  The property of $\cX$ required, which we call {\em Almost Order Continuous Equidistributional Average (AOCEA)}, asks that every nonnegative $X\in \cX$ possesses averages of equidistributed copies that are arbitrarily close to the order continuous part $\cX_a$ of $\cX$; see the subsection below for definition of $\cX_a$.
A thorough analysis of the AOCEA property is carried out in Section 2.  The results are then applied to prove the Main Theorem in Section 3.  Since for a  general law-invariant risk measure,  the  Fatou property is only valid  on a proper subset of $\cX$, topological lower semicontinuity and Fenchel-Moreau dual representation no longer follow directly from previously known results.
Nevertheless, in Section 4, under the AOCEA property, we recover the result that every convex, decreasing, law-invariant functional $\rho:\cX\to \R$ is $\sigma(\cX,\cX')$ lower semicontinuous at any $X\in \cX$ whose negative  tails  have vanishing norm (Theorem \ref{autocont}); see the subsection below for definition of $\cX'$.  Furthermore, the dual representation formula is valid at such $X$'s (Theorem \ref{dualrep}).

\subsection{Notation and Facts}
Throughout the paper, $(\Omega,\mathcal{F},\mathbb{P})$ stands for a non-atomic probability space. It is a standard fact  that  there exists a random vector on $\Omega$ with  any prescribed joint distribution. In particular, if $\mathbf{X}'$ is a random vector on a non-atomic probability space $(\Omega',\mathcal{F}',\mathbb{P}')$, then there exists a random vector $\mathbf{X}$ on $(\Omega,\mathcal{F},\mathbb{P})$ that has the same distribution as $\mathbf{X}'$, which we write as $\mathbf{X}\sim \mathbf{X}'$. These notions and facts   extend in a plain manner to finite measure spaces of the same measure. In particular, we will often consider a set $A\in\mathcal{F}$ endowed with the   probability structure restricted from $(\Omega,\mathcal{F},\mathbb{P})$. Given a (measurable) partition $\pi=(A_i)_{i\in I}$ of $\Omega$, where $I$ is at most countable, we often define random variables on $\Omega$ by specifying its values on each $A_i$ and then gluing the pieces together. A frequently used fact is as follows. Let $\pi'=(A_i')_{i\in I}$  be a partition of $\Omega'$  for another probability space $(\Omega',\mathcal{F}',\mathbb{P}')$, such that $\bP(A_i)=\bP'(A_i')$ for each $i\in I$. If $\mathbf{X}$ and $\mathbf{X}'$ are random vectors on $\Omega$ and $\Omega'$, respectively, such that $\mathbf{X}\vert{A_i}\sim \mathbf{X}'\vert{A_i'}$ for any $i\in I$, then $\mathbf{X}\sim \mathbf{X}'$.

\smallskip

Let $L^0:=L^0(\Omega)$ be the space of all random variables on $(\Omega,\mathcal{F},\mathbb{P})$  (to be precise, $L^0$ consists of equivalent classes of random variables modulo a.s.\ equality). Throughout the paper, $\cX$ stands for a {\em rearrangement-invariant (abbr., {r.i.}) space} over $(\Omega,\mathcal{F},\mathbb{P})$, other than $L^\infty$.  By an r.i.\ space, we mean that $\cX\neq\{0\}$ and $\cX$ is a Banach space of random variables in $L^0$ such that for any $X\in \cX$,
\begin{enumerate}
    \item if $Y\in L^0$ and $|Y|\leq |X|$ a.s., then $Y\in \cX$ and $\|Y\|\leq \|X\|$,
    \item if $Z\in L^0$  and $Z\sim X$, then $Z\in \cX$ and $\|Z\|=\|X\|$.
\end{enumerate}
The classical $L^p$-spaces  are all r.i.\ spaces. It is well known (see, e.g., \cite[Appendix]{CGX}) that $L^\infty\subset \cX\subset L^1$ and there exists a constant $C>0$ such that
\begin{align}\label{l1norm}
    \norm{X}_{L^1}\leq C\norm{X}\quad \text{ for every }X\in\cX.
\end{align}

Besides the definition,   two  notions about r.i.\ spaces are needed for the paper. Given an r.i.\ space $\cX$, its \emph{associate space} $\cX'$   is defined by
$$\cX'=\{Y\in L^0: \;\;\E[\abs{XY}]<\infty\;\text{ for any }X\in \cX\}.$$
$\cX'$ itself is also an r.i.\ space and it sits naturally as a closed subspace in the norm continuous dual $\cX^*$ of $\cX$. In the literature, dual representations with respect to $\cX'$ are regarded as ``tractable"  (cf., e.g., \cite{BC,BCb}). In the Banach lattice literature, $\cX'$ is just the order continuous dual $\cX_n^\sim$.
The \emph{order continuous part} $\cX_a$ of $\cX$ is defined by
\begin{align}\label{ocpdefn}
    \cX_a:=\Big\{X\in \cX:\;\;\lim_{\bP(A)\rightarrow0}\norm{X\mathbf{1}_A}=0\Big\}.
\end{align}
The members in $\cX_a$ are termed as order continuous in $\cX$, or as having absolutely continuous norm in some literature.  
 An r.i.\ space $\cX$ is said to be \emph{order continuous} if $\cX_a=\cX$, or equivalently, if $\cX'=\cX^*$ (\cite[Theorem 2.4.2]{MN}).
It is well known that $(L^p)'=L^q$ for any $ p,q\in[1, \infty]$ with $\frac{1}{p}+\frac{1}{q}=1$ and that $L^p (1\leq p<\infty)$ is order continuous but $L^\infty$ is not. In fact, $(L^\infty)_a=\{0\}$.

We exclude $L^\infty$ from our consideration for two main reasons. Firstly, the results for $L^\infty$ are already established in \cite{Joui06}. Secondly, the results for $L^\infty$ and for other r.i.\ spaces hold for dramatically  different reasons, which we briefly point out now.
When $\cX\neq L^\infty$, it is known that $$L^\infty\subset \overline{L^\infty}=\cX_a\subset \cX\subset L^1,$$ where the closure of $L^\infty$ is taken with respect to the norm of $\cX$.
Therefore by \eqref{ocpdefn}, for any $X\in L^\infty$, $\lim_{\bP(A)\rightarrow0}\norm{X\mathbf{1}_A}=0$. This fact, which clearly fails in $L^\infty$, will  serve as a primary tool in our developments. It is the essential technical difference between our model space $\cX $ and $L^\infty $. Finally, we remark  that since $\cX\neq L^\infty$,
\begin{align}\label{ocheart}
    X\in\cX_a \iff \lim_n\norm{X\mathbf{1}_{\{\abs{X}\geq n\}}}=0.
\end{align}

\smallskip

A function $\Phi:[0,\infty)\to[0,\infty)$ is called an Orlicz function if it is convex, increasing, non-constant, and $\Phi(0)=0$.   The Orlicz space $L^\Phi:=L^\Phi(\Omega)$ is the space of all $X\in L^0$ such that
\[
\norm{X}_{\Phi}:=\inf\left\{\lambda>0:\E\left[\Phi\left(\frac{\abs{X}}{\lambda}\right)\right]\leq 1 \right\}<\infty.
\]
$L^\Phi$ with the Luxemburg norm $\|\cdot\|_\Phi$ is an r.i.\ space. Furthermore, $(L^\Phi)'=L^{\Psi}$ if $\lim_{t\rightarrow\infty}\frac{\Phi(t)}{t}=\infty$ and $(L^\Phi)'=L^\infty$ otherwise, where $\Psi$ is the conjugate of $\Phi$ given by $$\Psi(s)=\sup\{st-\Phi(t):t\geq0\}\text{ for all }s\geq0.$$
The order continuous part of $L^\Phi$ is given by the space of all $X\in L^\Phi$ such that
\[
\E\left[\Phi\left(\frac{\abs{X}}{\lambda}\right)\right]<\infty\;\;\text{
for all }\lambda>0.\] This space is also called the Orlicz heart of $L^\Phi$ and is denoted by $H^\Phi$. It is known that $L^\Phi$ is order continuous (i.e., $L^\Phi=H^\Phi$) iff  $\Phi$ satisfies the $\Delta_2$-condition, i.e., there exist $C>0$ and $t_0>0$ such that
\[
\Phi(2t)\leq C\Phi(t),\quad \forall \ t>t_0.
\]

\smallskip

A functional $\rho:\cX\rightarrow(-\infty,\infty]$ is called a coherent risk measure if it is
\begin{enumerate}
    \item decreasing, i.e., $\rho(X_1)\leq \rho(X_2)$ whenever $X_1,X_2\in\cX$ satisfies $X_1\geq X_2$,
    \item subadditive, i.e., $\rho(X_1+X_2)\leq \rho(X_1)+\rho(X_2)$ for any $X_1,X_2\in\cX$,
    \item positive homogeneous, i.e., $\rho(\lambda X)=\lambda \rho(X)$ for any $X\in\cX$ and any real number $\lambda\geq 0$,
    \item cash invariant, i.e., $\rho(X+m\mathbf{1})=\rho(X)-m$ for any $X\in\cX$ and $m\in\R$.
\end{enumerate}
A functional $\rho:\cX\rightarrow(-\infty,\infty]$ is said to be convex if $\rho\big(\lambda X_1+(1-\lambda)X_2\big)\leq \lambda \rho(X_1)+(1-\lambda)\rho(X_2)$ for any $X_1,X_2\in\cX$ and $\lambda\in[0,1]$. $\rho$ is said to be law invariant if $\rho(X_1)=\rho(X_2)$ for any $X_1,X_2\in\cX$ with $X_1\sim X_2$.

Given any $X\in\cX$, a functional $\rho:\cX\rightarrow(-\infty,\infty]$
is said to have the Fatou property at $X$ if $\rho(X)\leq \liminf_n\rho(X_n)$ for any sequence $(X_n)$ that order converges to $X$ in $\cX$. By \emph{order convergence} in $\cX$, we mean that $X_n\stackrel{a.s.}{\longrightarrow}X$ and there exists $X_0\in\cX$ such that $\abs{X_n}\leq X_0$ for all $n\in\N$. $\rho$ is said to be $\sigma(\cX,\cX')$ lower semicontinuous at $X$ if   $\rho(X)\leq \liminf_\alpha\rho(X_\alpha)$ for any net $(X_\alpha)$ that converges to $X$ in $\sigma(\cX,\cX')$, or equivalently, if $\{\rho>\lambda\}$ is a $\sigma(\cX,\cX')$-neighborhood of $X$ for any real number $\lambda$ such that $\rho(X)>\lambda$.

 \smallskip

Finally, the convex hull of a set $\mathcal{A}\subset \cX$ is denoted by $\mathrm{co}(\mathcal{A})$. The distance of $X\in\cX$ and $\mathcal{A}\subset \cX$ is given by $\mathrm{d}(X,\mathcal{A})=\inf_{Y\in\mathcal{A}}\norm{X-Y}$; the distance of $\mathcal{A},\mathcal{B}\subset \cX$ is given by $\mathrm{d}(\mathcal{A},\mathcal{B})=\inf_{X\in\mathcal{A},Y\in\mathcal{B}}\norm{X-Y}$.
The positive and negative parts of $X\in\cX$ are given by $X^+:=\max\{X,0\}$ and $X^-:=\max\{-X,0\}$, respectively.

 \smallskip

We refer the reader  to \cite{BS88} for facts and results on general r.i.\ spaces, to \cite{EdgarSucheston1992} for detailed information on Orlicz spaces, and to \cite{AB,MN} for relevant terminology and facts on Banach lattices and order structures.

\section{The Main Result. Automatic Fatou Property}
\subsection{Formulation of the Main Result}
We begin with  the following example, which indicates that for  the class of real-valued, law-invariant,  coherent risk measures,  automatic Fatou property cannot be expected at any random variable $X\in\cX$ such that $X^-\not\in \cX_a$. Note that $X^-\notin\cX_a\iff \lim_n\norm{X\one_{\{X\leq -n\}}}>0$.
\begin{example}\label{counterexam}
Let $\cX$ be an r.i.\ space over a non-atomic probability space such that $\cX\not= L^\infty$.
Consider the functional $\rho:\cX\to \R$ given by
\begin{align}\label{firstexample}
\rho(X) = \mathrm{d}(X^-,\cX_a)-\E[X]=\inf_{Y\in\cX_a}\norm{X^--Y}-\E[X].
\end{align}
We show that $\rho$ is a law-invariant coherent risk measure that fails the Fatou property at any $X\in\cX$ with $X^- \notin \cX_a$.

\begin{enumerate}
    \item $\rho$ is a coherent risk measure.\\
    Since $\cX_a$ is a norm-closed subspace in $\cX$, the quotient space $\cX/\cX_a$ is a Banach space with the quotient norm $\norm{[X]}_q:=\mathrm{d}(X,\cX_a)$, where $[X]$ is the equivalent class of $X\in\cX$ in $\cX/\cX_a$; see, e.g., \cite[Theorem 1.11]{AA06}. Moreover, since $\cX_a$ is an order ideal in $\cX$ (i.e., if $\abs{Y}\leq \abs{X}$ and $X\in\cX_a$ then $Y\in\cX_a$), $\cX/\cX_a$ with the quotient norm $\norm{\cdot}_q$ and the quotient order $[X]\vee [Y]:=[X\vee Y]$ is in fact a Banach lattice; see, e.g., \cite[p.\ 3]{LZ}. With these observations, we write 
    \begin{align}\label{simpleformd}
    \rho(X)=\norm{[X^-]}_q-\E[X].\end{align}
Using \eqref{simpleformd}, it is clear that $\rho$ is decreasing and positive homogeneous. Let's show subadditivity of $\rho$. Take any $X_1,X_2\in\cX$. Since $0\leq (X_1+X_2)^-\leq X_1^-+X_2^-$, $$[0]\leq [(X_1+X_2)^-]\leq [X_1^-+X_2^-]=[X_1^-]+[X_2^-]$$ in the quotient space $\cX/\cX_a$. Therefore, $$\bignorm{[(X_1+X_2)^-]}_q\leq \bignorm{[X_1^-]+[X_2^-]}_q\leq \bignorm{[X_1^-]}_q +\bignorm{[X_2^-]}_q.$$
From this it follows easily that $\rho(X_1+X_2)\leq \rho(X_1)+\rho(X_2)$.  To complete
the proof that $\rho$ is a coherent risk measure, it remains to show cash invariance
of $\rho$. Take any $X\in\cX$ and $m\in\R$. Then $(X+m\one)^--X^-\in L^\infty\subset \cX_a$. Hence    $$\bignorm{[(X+m\one)^-]}_q=\bignorm{[X^-]}_q.$$
It follows easily from \eqref{simpleformd} that   $\rho(X+m\one)=\rho(X)-m$.
\item $\rho$ is law  invariant. \\
For a real number $r\geq 0$, set $h(x)=x^--x^-\wedge r$. If $X'\sim X$, then $h(X')\sim h(X)$ and $\norm{h(X')}=\norm{h(X)}$. Thus for law invariance of $\rho$, it is enough to show that for any $X\in\cX$,
\begin{align}\label{li}
\rho(X)=\inf_{r\in\R,r\geq 0}\norm{X^--X^-\wedge r\mathbf{1} }-\E[X].
\end{align}
Denote the   right hand side of \eqref{li} by $\widetilde{\rho}(X)$.  Since $0\leq X^-\wedge r\mathbf{1} \in L^\infty\subset \cX_a$, by comparing  the infima in \eqref{firstexample} and \eqref{li}, it is immediate that $\rho(X)\leq \widetilde{\rho}(X)$.
On the other hand, recall that $L^\infty $ is norm dense in   $\cX_a$. Thus it is easy to see from \eqref{firstexample} that for any $X\in \cX$, $$\rho(X) = \mathrm{d}(X^-,L^\infty)-\E[X]=\inf_{Y\in L^\infty}\norm{X^--Y}-\E[X].  $$
For any $Y\in L^\infty$, since $Y\leq \norm{Y}_\infty \mathbf{1}$, $\abs{X^--Y}\geq (X^--Y)^+\geq (X^--\norm{Y}_\infty \mathbf{1})^+=X^--X^-\wedge (\norm{Y}_\infty \mathbf{1})$. Thus \begin{align*}
    \norm{X^--Y}-\E[X]\geq   \bignorm{X^--X^-\wedge (\norm{Y}_\infty \mathbf{1})}-\E[X]\geq \widetilde{\rho}(X).
\end{align*}
Taking infimum over $Y\in L^\infty$ yields $\rho(X)\geq \widetilde{\rho}(X)$. It follows that $\rho(X)=\widetilde{\rho}(X)$.
\item $\rho$ fails the Fatou property at every $X\in\cX$ with $X^-\notin\cX_a$.\\
Let $X\in\cX$ be such that $X^-\not\in \cX_a$. Set $X_n=(X\vee (-n\mathbf{1}))\wedge n\mathbf{1}$ for $n\in\N$. Then $ X_n\stackrel{a.s.}{\longrightarrow}X$ and $\abs{X_n}\leq \abs{X}$ for any $n\in\N$. It suffices to  show that $\rho(X)>\lim_n\rho(X_n)$. Since $X_n\in L^\infty \subset \cX_a$,  $\mathrm{d}(X_n,\cX_a)=0$. By Dominated Convergence Theorem, $$\rho(X_n)=\mathrm{d}(X_n,\cX_a)-\E[X_n]=-\E[X_n]\rightarrow -\E[X].$$
However, as $X^-\not\in \cX_a$, $\rho(X)=\mathrm{d}(X^-,\cX_a)-\E[X]>-\E[X]$. Therefore, $\rho(X)>\lim_n\rho(X_n)$, as required.
\end{enumerate}
\end{example}

Example \ref{counterexam} tells us that, in looking for random variables $X$ in an r.i.\ space at which all real-valued, law-invariant, coherent risk measures automatically satisfies the Fatou property, one must confine the search to those $X$'s with $X^- \in \cX_a$.
Remarkably, in most classical r.i.\ spaces, all real-valued, law-invariant, coherent risk measures are indeed automatically Fatou at all such $X$.    This is the case, for instance, in all $L^p$ spaces, $1\leq p\leq \infty$, all Orlicz spaces and Orlicz hearts and all order continuous r.i.\ spaces. In fact,  the precise
structural property on an r.i.\ space can be identified in order for this to happen. Let $\cX$ be an r.i.\ space over a non-atomic probability space other than $L^\infty$. For $X\in \cX$, define
\[ \mathcal{CL}(X) = \mathrm{co}\{Y:Y \sim X\}.\]
We say that $\cX$ has the \textbf{Almost Order Continuous Equidistributional Average (abbr., AOCEA) property} 
if for any $X\in \cX_+$, $$\mathrm{d}(\mathcal{CL}(X),\cX_a) =0.$$  Note that in the definition of the AOCEA property, one may replace the set $\mathcal{CL}(X)$ with 
the set 
\begin{align}\label{alset}
    \mathcal{AL}(X) = \Big\{\frac{1}{n}\sum^n_{k=1}X_k: n\in \N, X_1,\dots, X_n \sim X\Big\}.
\end{align} 
This follows from the observation that 
 the set of convex combinations with rational coefficients of elements of the set $\{Y: Y \sim X\}$ is norm dense in $\mathcal{CL}(X)$ and hence so is the set $\mathcal{AL}(X)$. 
Therefore, the AOCEA property says that every   nonnegative random variable in $\cX$ possesses averages of equidistributed copies that are almost order continuous, i.e., arbitrarily close to $\cX_a$.

We can now state the   main result of the paper.
\begin{theorem}\label{main-thm}
Let $\cX$ be an r.i.\ space over a non-atomic probability space $(\Omega,\mathcal{F},\bP)$ other than $L^\infty$.
The following statements are equivalent:
\begin{enumerate}
       \item  Every law-invariant, coherent risk measure  $\rho:\cX\to \R$ has the Fatou property at $0$.
   \item Every convex, decreasing, law invariant functional  $\rho:\cX\to \R$ has the Fatou property at any $X\in \cX$ such that $X^-\in \cX_a$.
    \item $\cX$ satisfies the AOCEA property.
    \end{enumerate}
\end{theorem}

A detailed analysis of the AOCEA property is given in the next subsection.  Proof of Theorem \ref{main-thm} will be presented in Section 3.

\subsection{An in-depth look at  the AOCEA property}\label{condition2}
In this part, we   provide a detailed investigation of the AOCEA property.  The outcome of the investigation will also lend significant aid to the proof of Theorem 2.2 itself.
The main aspects of the property are revealed in the following proposition.
Given a sequence $(A_n)_{n=1}^\infty$ of measurable sets, we write $A_n\downarrow \emptyset$ if $A_n \supseteq A_{n+1}$ for all $n\geq 1$ and $\bigcap_{n=1}^\infty A_n = \emptyset$.

\begin{proposition}
\label{MainLemma}
Let $\cX$ be an r.i.\ space over a non-atomic probability space $(\Omega,\mathcal{F},\bP)$ other than $L^\infty$. The following statements are equivalent:
 \begin{enumerate}
 \item $\cX$ satisfies the AOCEA property, i.e., for any $X\in \cX_+ $,
\[ \mathrm{d}(\mathcal{CL}(X),\cX_a) =0, \text{ where } \mathcal{CL}(X) = \mathrm{co}\{Y:Y \sim X\}.\]
\item For  any $A\in\mathcal{F}$ with $\bP(A)>0$, any $X\in \cX$ supported in $A$,  and any $\varepsilon>0$, there exist random variables $(X_i)_{i=1}^k$, all supported in $A$, with $X_i\sim X$ for $i=1,\dots,k$, a convex combination $\sum_{i=1}^k\lambda_iX_i$, and $V\in\cX_a$, also supported in $A$, such that $$\Bignorm{\sum_{i=1}^k\lambda_iX_i-V}<\varepsilon.$$
\item For any $X \in \cX_+$,    any sequence of measurable sets $(A_n)_{n=1}^\infty$ with $A_n \downarrow \emptyset$, and any $\varepsilon >0$, there exist $n_1,\dots,n_k\in\N$, random variables $(Z_i)_{i=1}^k$ and a convex combination $Z=\sum_{i=1}^k\lambda_i Z_i$ such that   $Z_i \sim X\mathbf{1}_{A_{n_i}}$    for $i=1,\dots,k$ and $\norm{Z} <\varepsilon$.
\item  For any $X \in \cX$, any $A\in\mathcal{F}$ with $\bP(A)>0$,  any sequence of measurable sets $(A_n)_{n=1}^\infty$ with $A_n \downarrow \emptyset$, and any $\varepsilon >0$, there exist $n_1,\dots,n_k\in\N$, random variables $(Z_i)_{i=1}^k$ and a convex combination $Z=\sum_{i=1}^k\lambda_i Z_i$ such that  $Z_i \sim X\mathbf{1}_{A_{n_i}}$    for $i=1,\dots,k$, all $Z_i$'s are supported in $A$, and $\norm{Z} <\varepsilon$.
\end{enumerate}
\end{proposition}

(2) is the ``localized" version of (1). In particular, taking $A=\Omega$ in $(2)$ yields (1) for all random variables, not necessarily nonnegative. Similarly, one can compare (3) and (4). The four equivalent formulations are each useful in their own way.
(1) is succinct
and aesthetically pleasing;  (3) is easier to verify in practice.  (2) will be used for establishing automatic $\sigma(\cX,\cX')$ lower semicontinuity in Subsection \ref{autocontinuity};  (4) will be used in the proof of Theorem \ref{main-thm} (2)$\implies$(1) in Subsection \ref{2to1}.

The proof of this proposition is, however, rather involved. We put it into Appendix \ref{proofAOCEA} in order to facilitate the accessibility of the main results of the paper on automatic Fatou property and tractable dual representations of law-invariant risk measures.

\medskip

The AOCEA property is satisfied by most classical r.i.\ spaces. 
First of all, it trivially holds for an order continuous r.i.\ space $\cX$  because  $\cX = \cX_a$.  Therefore, Lebesgue spaces $L^p $ $(1\leq p<\infty)$ 
and Orlicz hearts all satisfy the property.
The next proposition shows that Orlicz spaces, which have been widely used as model spaces in the recent literature, also satisfy the property.  

\begin{proposition}\label{orlicznice}
An Orlicz space $L^\Phi$ has the AOCEA property.
\end{proposition}

\begin{proof} 
We verify (3) of Proposition \ref{MainLemma}.
Let $X\in  (L^\Phi)_+$, $A_n\downarrow \emptyset$ and $\varepsilon>0$ be given. Since $\bP(A_n)\rightarrow0$, by passing to a subsequence, we may assume that $\sum_{n=1}^\infty \bP(A_n)<1$. By non-atomicity of $\bP$,  there exists a disjoint sequence  $(B_n)_{n=1}^\infty$ of measurable sets such that $\bP(B_n)=\bP(A_n)$ for any $n\in\N$. Using non-atomicity again, we obtain, for each $n\in\N$,  a random variable $Z_n$, supported in $B_n$, such that  $Z_n\sim X\mathbf{1}_{A_n}$.

Let $\eta >0$ be such that $\E[\Phi(\eta  {X})] <\infty$. Choose $k\in \N$ so that $k\eta \varepsilon \geq 1$.
By Dominated Convergence Theorem, $\lim_n \E[\Phi(\eta {X}\mathbf{1}_{A_n})] = 0$. Thus we can pick $n_1,\dots, n_k$ such that
\[ \sum^k_{i=1}\E[\Phi(\eta  {X}\mathbf{1}_{A_{n_i}})]\leq 1.\]
Set $Z = \frac{1}{k}\sum^k_{i=1}Z_{n_i}$.  Then
\begin{align*}
\E\left[\Phi\Big(\frac{ {Z}}{\varepsilon}\Big)\right] & = \E\left[\Phi\Big(\frac{\eta  {Z}}{\eta\varepsilon}\Big)\right]\\& =
\sum^k_{i=1}\E\left[\Phi\Big(\frac{\eta  {Z_{n_i}}}{k\eta\varepsilon }\Big)\right] \text{ (since $Z_{n_i}$'s are disjoint)}\\
 & \leq \sum^k_{i=1}\E[\Phi(\eta   {X}\mathbf{1}_{A_{n_i}})]\leq  1.
\end{align*}
Hence, $\|Z\| \leq \varepsilon$. This completes the verification of condition (3) of Proposition \ref{MainLemma}.
\end{proof}
Combining Proposition~\ref{orlicznice} with Theorem~\ref{main-thm}, we obtain the following result.
\begin{corollary}\label{orliczfatou}
A real-valued, convex, decreasing, law-invariant functional  on an Orlicz space $L^\Phi$  has the Fatou property at any $X\in L^\Phi$ such that $X^-\in H^\Phi$.
\end{corollary}

It should be noted that the AOCEA property is not universally satisfied by all r.i.\
spaces. A brief example of an r.i.\ space failing it is presented at Appendix~\ref{poorexample}. 
\medskip

We also refer to Chen et al \cite{CGLLb} for other interesting applications of the AOCEA property, e.g., regarding collapse to the mean of law-invariant linear functionals. See also Bellini et al \cite{BCb} for more results on collapse to the mean. 

\medskip

At this point, we make a slight digression to discuss the set $\mathcal{CL}(X) =\mathrm{co}\{Y:Y\sim X\}$ that appears in the AOCEA property.
It has been   drawing attention in some recent works. 
Its relation with the well-studied set $\{Y\in \cX: Y\preceq_{cx} X\}$ is investigated in Bellini et al \cite{BC}. Recall that for random variables $X, Y \in L^1$, $Y \preceq_{cx} X$ means that $\E[f(Y)] \leq \E[f(X)]$  for every convex function $f: \R\to \R$, whenever the expectations exist. It was proved in \cite{BC} that
$$ \ol{\mathcal{CL}(X)}^{\sigma(\cX,\cX')}=\{Y\in \cX: Y\preceq_{cx} Y\},$$
where the left hand side is the closure of $\mathcal{CL}(X)$ in the $\sigma(\cX,\cX')$ weak topology.
This identity  implies in particular the coincidence of  law invariance and Schur convexity for certain convex functionals; see \cite{BC} for details. More results and applications about the set $\mathcal{CL}(X)$ can also be  found in Gao et al \cite{GLL}. 
Sets of a similar fashion as the set $\mathcal{AL}(X)$ in \eqref{alset} has also appeared in Mao and Wang \cite{MW} in their study of risk aggregation. We believe that more appealing aspects and applications of the set $\mathcal{CL}(X)$ will come forth in the literature.

\section{Proof of the Main Result}
In this section, we prove the main result of the paper, Theorem \ref{main-thm}. The implication (2)$\implies $(1) is obvious. In the following two subsections, we prove the implications (3)$\implies$(2) and (1)$\implies$(3), respectively.

\subsection{Proof of (3)$\implies$(2) in Theorem \ref{main-thm}}\label{2to1}

The following technical lemma reduces the Fatou property to a much simpler form.
\begin{lemma}\label{techlem2}
Let $\cX$ be an r.i.\ space over a probability space.
Let $\rho:\cX\rightarrow(-\infty,\infty]$ be a decreasing functional and $X\in\cX$. The following are equivalent:
\begin{enumerate}
    \item $\rho$ has the Fatou property at $X$;
    \item Let  $Y\in\cX$ be such that $Y\geq X$, $(c_n)$ be real numbers such that $c_n\downarrow 0$, and $(A_n)$ be measurable sets such that $A_n\downarrow\emptyset$. For any $n\in\N$, put $Y_n=X\mathbf{1}_{A_n^c}+c_n\mathbf{1}_{A_n^c}+Y\mathbf{1}_{A_n}$. Then $\rho(X)=\lim_n\rho(Y_n)$.
\end{enumerate}
\end{lemma}
\begin{proof}
Assume that (1) holds. Let  $Y $,   $(c_n)$,  $(A_n)$ and $(Y_n)$ be as given in (2). Since $Y_n\geq X\mathbf{1}_{A_n^c}+0 +X\mathbf{1}_{A_n}=X$, $\rho(Y_n)\leq \rho(X)$ for any $n\geq 1$. In particular, $\limsup_n\rho(Y_n)\leq \rho(X)$. Since $X\leq Y_n\leq Y\mathbf{1}_{A_n^c}+c_1\mathbf{1}+Y\mathbf{1}_{A_n}=Y+c_1\mathbf{1}$, $(Y_n)$ is dominated in $\cX$. It is easy to see that $Y_n\stackrel{a.s.}{\longrightarrow}X$.  Thus by the assumption (1), $\rho(X)\leq \liminf_n\rho(Y_n)$. It follows that $\rho(X)=\lim_n\rho(Y_n)$. This proves that (1)$\implies$(2).

Assume that (2) holds. Suppose otherwise that (1) fails. Then  there exists a sequence $(X_n)$ that order converges to $X$ in $\cX$ but $\rho(X)>\liminf_n\rho(X_n):=\lambda\in[-\infty,\infty)$. By switching to a subsequence of $(X_n)$, we may assume that \begin{align}\label{smalll1}
    \rho(X_n) \to \lambda < \rho(X).
\end{align}

For any $k\in\N$, by Egoroff's Theorem, there exist $n_k\in\N$ and a measurable set $B_k$ such that $\bP(B_k)\leq \frac{1}{2^k}$ and $\abs{X_{n_k }-X}\leq\frac{1}{k}$ on $B_k^c$. Set $A_k=\bigcup_{m\geq k}B_m$. Then $(A_k)$ is a decreasing sequence of measurable sets such that  $\bP(A_k)\leq \sum_{m\geq k}\bP(B_m)\leq \sum_{m\geq k}\frac{1}{2^m}=\frac{1}{2^{k-1}}\rightarrow 0$ and thus $\bP(\bigcap_{k=1}^\infty A_k)=0$. Without loss of generality, we may assume that
$A_k\downarrow\emptyset$. Clearly, $\abs{X_{n_k}-X}\leq\frac{1}{k}$ on $A_k^c$ for any $k\geq 1$. Since $(X_n)$ order converges to $X$ in $\cX$,  there exists $Y\in \cX$ such that $X_n \leq Y$ for all $n\in\N$. Since $X_n\stackrel{a.s.}{\longrightarrow}X$, $X\leq Y$.
Set
$ Y_k = X\mathbf{1}_{A_k^c}+\frac{1}{k}\mathbf{1}_{A_k^c}+Y\mathbf{1}_{A_k}$ for $k\in\N$. Then $$X_{n_k}=X_{n_k}\mathbf{1}_{A_k^c}+X_{n_k}\mathbf{1}_{A_k} \leq \Big(X+\frac{1}{k}\Big)\mathbf{1}_{A_k^c}+Y\mathbf{1}_{A_k}= Y_k$$ and hence $\rho(Y_k) \leq \rho(X_{n_k})$. By assumption (2),
\[ \rho(X) = \lim_k \rho(Y_k) \leq \lim_k \rho(X_{n_k}) =\lim_n\rho(X_n) =\lambda,\]
contradicting \eqref{smalll1}.  This proves that (2)$\implies$(1).
\end{proof}

\begin{proof}[Proof of Theorem \ref{main-thm} (3)$\implies$(2)] Assume that $\cX$ satisfies the AOCEA property. Let $\rho:\cX\to \R$ be a convex, decreasing, law-invariant functional. Let $X\in \cX$ be such that $X^-\in \cX_a$. We   apply the previous lemma  to establish the Fatou property at $X$. Suppose otherwise that (2) in Lemma \ref{techlem2}  fails.
Then there exist $ \cX\ni Y\geq X$, $c_n\downarrow 0$ and  $A_n\downarrow\emptyset$ such that $\rho(Y_n)\not\rightarrow\rho(X)$, where $Y_n =( X+c_n\mathbf{1})\mathbf{1}_{A_n^c}+ Y\mathbf{1}_{A_n}$ for $n\in\N$.  Since $Y_n\geq X$, $\rho(Y_n)\leq \rho(X)$.  Thus there exist $\varepsilon>0$ and a subsequence $(Y_{n_k})$ such that $\rho(Y_{n_k})<\rho(X)- \varepsilon$ for all $k$. Replacing $(Y_n)$ with $(Y_{n_k})$, we may assume that
\begin{align}\label{little}
    \rho(Y_n)<\rho(X)-\varepsilon\quad\text{ for all }n\in\N.
\end{align}

We aim at contradicting \eqref{little}.
Recall from  \cite[Proposition 3.1]{RS06} that $\rho$ is norm continuous. Thus there exists $\delta >0$ such that \begin{align}\label{normcon}
    \rho(X+V)>  \rho(X)- \varepsilon\quad\text{ if }V\in \cX\text{ and } \norm{V}\leq\delta.
\end{align}
Take a real number $r>0$ such that  $\mathbb{P}(B) >0$, where $B= \{|X| \leq r\}$.
Since $\mathbf{1}\in L^\infty\subset  \cX_a$, there exists $\eta>0$ such that
\begin{align}\label{estimate1}
    \norm{\mathbf{1}_C} \leq \frac{\delta}{6r}\quad\text{ if }\mathbb{P}(C) \leq  \eta.
\end{align}Similarly, since $X^-\in \cX_a$ and $\mathbb{P}(A_n) \to 0$, $\|X^-\mathbf{1}_{A_n}\| \to 0$.
Thus by passing to subsequences if necessary, we may assume that $\bP(A_1)<\bP(B)$ and  $$\sum_{n=1}^\infty \mathbb{P}(A_n) \leq  \min\big\{\eta, \mathbb{P}(B\backslash A_1)\big\},$$
$$\|X^-\mathbf{1}_{A_n}\|\leq  \frac{\delta}{6}\quad\text{ for all }n\in\N,$$
$$c_n\|\mathbf{1}\| \leq  \frac{\delta}{6}\quad\text{ for all }n\in\N.$$

By Proposition \ref{MainLemma}(4) and passing to a subsequence of $(A_n)$ if necessary, there exist random variables  $(Z_n)_{n=1}^k$, all supported in $B\backslash A_1$, and a convex combination $  \sum^k_{n=1}\lambda_nZ_n $ such that  $$Z_n \sim Y\mathbf{1}_{A_n} \text{ for }n=1,\dots, k,\quad \text{and}\quad\Bignorm{\sum^k_{n=1}\lambda_nZ_n} < \frac{\delta}{6}.$$
For $n=1,\dots,m$, since $\{Z_n\neq 0\}\subset B\backslash A_1$ and $A_n\subset A_1$, $\{Z_n\neq 0\}\cap A_n=\emptyset$. Moreover, $\bP(Z_n\neq 0)=\bP(Y\mathbf{1}_{A_n}\neq 0)\leq \bP(A_n)$. Take $B_n\subset A_n$ such that $\bP(B_n)=\bP(A_n)-\bP(Z_n\neq0)$.
Hence, we get a partition of $\Omega$:
\begin{align}\label{newpart}
\Omega=\big(A_n^c\cap \{Z_n=0\}\big)\cup \big(A_n\backslash B_n\big)\cup \big(B_n\cup \{Z_n\neq 0\}\big).
\end{align}
On the other hand, we have another partition of $\Omega$ as follows:
\begin{align}\label{oldpart}
    \Omega=\big(A_n^c\cap \{Z_n=0\}\big)\cup   \{Z_n\neq 0\} \cup A_n.
\end{align}Since $\{Z_n\neq 0\}$ has the same probability as $A_n\backslash B_n$, we can take a random variable $W_n$, supported on $A_n\backslash B_n$, such that $W_n\sim (X+c_n\mathbf{1})\mathbf{1}_{  \{Z_n\neq0\}}$.
Now for any $n=1,\dots,k$, define $Y_n'$ piecewise according to the partition of $\Omega$ in \eqref{newpart}:
$$Y_n'=(X+c_n\mathbf{1})\mathbf{1}_{A_n^c\cap \{Z_n=0\}}+W_n+Z_n.$$
Comparing with $Y_n$ along the partition in \eqref{oldpart}, we see that $Y_n'\sim Y_n$.

Recall that $\{Z_n\neq 0\}\subset B\backslash A_1\subset B$. Thus by the choice of $B$, $\abs{X\mathbf{1}_{\{Z_n\neq 0\}}}\leq r\mathbf{1}_{\{Z_n\neq0\}}$. By \eqref{estimate1} and the fact that $\bP(Z_n\neq 0)\leq \bP(A_n)\leq \eta$, we have $$\norm{X\mathbf{1}_{\{Z_n\neq0\}}}\leq r\norm{\mathbf{1}_{\{Z_n\neq 0\}}}\leq r\frac{\delta}{6r}=\frac{\delta}{6}.$$
It follows that
\[\norm{W_n} \leq \norm{X\mathbf{1}_{\{Z_n\neq0\}}} + c_n\norm{\mathbf{1}} \leq \frac{\delta}{3}\quad \text{  for }n=1,\dots,k.\]
Recall that  $\{Z_n\neq 0\}$ and $ A_n$ are disjoint. Thus  rewrite
\[
Y_{n}'= X- X^+\mathbf{1}_{A_{n}}+  X^-\mathbf{1}_{A_{n}}-X\mathbf{1}_{\{Z_n\neq 0\}} + c_{n} \mathbf{1}_{A_{n}^c\cap \{Z_n=0\}}+ W_{n}+Z_{n} .\]
Then
\[\sum^k_{n=1}\lambda_nY_n' = X -\sum^k_{n=1}\lambda_nX^+\mathbf{1}_{A_{n}}+V,
\]
where
\[ V=\sum^k_{n=1}\lambda_n(X^-\mathbf{1}_{A_{n}}-X\mathbf{1}_{\{Z_n\neq0\}} + c_{n} \mathbf{1}_{A_{n}^c\cap \{Z_n=0\}}+ W_{n}+Z_{n}).\]
We have
\[ \norm{V} \leq \Bignorm{\sum^k_{n=1}\lambda_nZ_n} +\sum^k_{n=1}\lambda_n\big(\norm{X^-\mathbf{1}_{A_{n}}} + \norm{X\mathbf{1}_{\{Z_n\neq0\}}} + c_{n}\|\mathbf{1}\| +\norm{W_{n}})\leq \delta.
\]
Hence by  monotonicity of $\rho$ and \eqref{normcon},
\begin{align}\label{norm-lsc}
     \rho\Big(\sum^k_{n=1}\lambda_nY'_{n}\Big) \geq \rho(X+V) > \rho(X) - \varepsilon.
\end{align}

Finally,
\begin{align*}
\sum^k_{n=1}\lambda_n\rho(Y_{n}) & = \sum^k_{n=1}\lambda_n\rho(Y_{n}') \text{ (law invariance)}\\
&\geq \rho(\sum^k_{n=1}\lambda_nY_{n}') \text{ (convexity)}\\
& > \rho(X) -\varepsilon.
\end{align*}
Hence, there exists $n$ such that $\rho(Y_{n}) > \rho(X) - \varepsilon$, contradicting \eqref{little}.
\end{proof}

\begin{remark}
The reader may re-examine  the role of real-valuedness of $\rho$ in the proof of (3)$\implies$(2). It is only used to ensure norm lower semicontinuity of $\rho$ at $X$; see \eqref{normcon} and \eqref{norm-lsc}. Therefore, one sees that the following    statement is also equivalent to the AOCEA property when $\cX\neq L^\infty$ and can be added to Theorem \ref{main-thm}:
\begin{enumerate}
    \item[(1')] Every convex, decreasing, law-invariant functional  $\rho:\cX\to (-\infty,\infty]$ has the Fatou property at any $X\in \cX$, where $\rho$ is norm lower semicontinuous and   $X^-\in \cX_a$.
   \end{enumerate}

\end{remark}

\subsection{Proof of (1)$\implies$(3) in Theorem \ref{main-thm}}
Throughout this subsection, assume that $\cX$ is an r.i.\ space over a non-atomic probability space and  $\cX\neq L^\infty$.  If $\cX$ fails the AOCEA property, we aim to construct a law-invariant coherent risk measure $\rho : \cX \to \R$ that fails the Fatou property at $0$.

The following discretization lemma will be useful in the course of the construction. 
Let $f$ be a positive linear functional on $\cX$, i.e., $f(X)\geq 0$ for any $X\geq 0$. Clearly, $f$ is positive iff $f$ is increasing, i.e., $f(X_1)\geq f(X_2) $ whenever $X_1\geq X_2$. By  Birkhoff's Theorem (\cite[Theorem 4.3]{AB}), $f$ is bounded on $\cX$. Therefore, if $X'\sim X\in\cX$, then $f(X')\leq \norm{f}\norm{X'}=\norm{f}\norm{X}$.
It follows that for any $X\in\cX$,
$$\sup\{f(X'): X'\sim X\} \in\R.$$
\begin{lemma}\label{discretization}
Let $f$ be a positive   linear functional on $\cX$. 
For any $X\in\cX$,
\begin{equation}\label{e3.6}\sup\{f(X'): X'\sim X\}   =\sup\{f(Z):  Z\sim U\leq X,U\in\mathcal{X}, U\text{ is discrete}\}.\end{equation}
\end{lemma}

\begin{proof}
Take any $X\in\cX$. Denote the left and right hand sides of \eqref{e3.6} by $\phi_1(X)$ and $\phi_2(X)$, respectively. Also, put $$\phi_3(X)=\sup\{f(Z):  Z\sim U\leq X,U\in\mathcal{X}\}.$$
We first show that $\phi_1(X)=\phi_3(X)$.
If  $X'\sim X$, then $X'\sim X\leq X$, implying that $\phi_1(X)\leq\phi_3(X)$.
Let $Z\in\cX$ be such that $Z\sim U\leq X$ for some $U\in\cX$. By Lemma \ref{movingorder}, there exists a random variable $X'$ such that $Z\leq X'\sim X$. Clearly, $X'\in\cX$. Since $f$ is increasing, $f(Z)\leq f(X')\leq \phi_1(X)$. Taking supremum over $Z$, we obtain $\phi_3(X)\leq \phi_1(X)$. It follows that $\phi_1(X)=\phi_3(X)$, as desired. 

Apparently, $\phi_3(X)\geq \phi_2(X)$.
To see the reverse inequality,   take any random variables $Z,U\in\cX$  such that   $  Z\sim U \leq X$.  For any  real number   $a>1$, put $$Z'= \sum_{n\in \mathbb{Z}} a^n \mathbf{1}_{\{a^n <Z\leq a^{n+1}\}}-\sum_{n\in \mathbb{Z}} a^{n+1} \mathbf{1}_{\{-a^{n+1} \leq Z<- a^n\}},$$
$$U'= \sum_{n\in \mathbb{Z}} a^n \mathbf{1}_{\{a^n <U\leq a^{n+1}\}}-\sum_{n\in \mathbb{Z}} a^{n+1} \mathbf{1}_{\{-a^{n+1} \leq U<- a^n\}}.$$
Then $U'$ is discrete and $Z'\sim U' \leq U\leq X$. Moreover, $Z\geq Z'\geq \frac{1}{a}Z^+-aZ^-$, so that $Z'\in\mathcal{X}$.  Hence
$$\phi_2(X)    \geq f(Z') 
    \geq \frac{1}{a}f(Z^+)-af(Z^-).$$
Letting $a\downarrow 1$ and  taking supremum over $Z$, we obtain $\phi_2(X)\geq \phi_3(X)$. It follows that $\phi_2(X)=\phi_3(X)=\phi_1(X)$, completing the proof.
\end{proof}

\begin{lemma}\label{p3.4}
Let $f$ be a positive  linear functional on $\cX$.
Define $\phi:\cX\to \R$ by 
\[ \phi(X) = \sup\{f(X'): X'\sim X\}.\]
Then $\phi$ is  law invariant, increasing, positive homogeneous  and subadditive    on $\cX$. If in addition $f$ vanishes on $L^\infty$, then $\phi(X+m\one) = \phi(X)$  for all $X\in \cX$ and $m\in\R$.
\end{lemma}

\begin{proof}
 
Clearly $\phi$ is law invariant and positive homogeneous.
Suppose that  $X_1\leq X_2$ in $\cX$ and take any $X'\sim X_1$. 
By Lemma \ref{movingorder}, there exists a random variable $X_2'$ such that $X'\leq  X_2'\sim X_2$. Since $f$ is increasing,
$f(X') \leq f(X_2')\leq \phi(X_2).$
Taking supremum over $X'$ gives $\phi(X_1)\leq \phi(X_2)$.
This proves that $\phi$ is increasing.
Next,   let us show that $\phi$ is subadditive.
Consider any  any $U,V\in\mathcal{X}$.  By Lemma \ref{discretization}, it suffices to show that if  $X \in \cX$ is a discrete random variable such that $X\sim Z \leq U+V$, then $f(X) \leq \phi(U) + \phi(V)$.
Let $\{a_i\}_{i\in I}$ be the (at most countable) set of  values   $a\in\R$ such that $\mathbb{P}(X=a)>0$. Then
$\mathbb{P}(Z=a_i)=\mathbb{P}(X=a_i)$. On $\{X=a_i\}$, find $U'\vert_{\{X=a_i\}}\sim U\vert_{\{Z=a_i\}}$, and put$$ V'\vert_{\{X=a_i\}}:=a_i -U'\vert_{\{X=a_i\}}\sim a_i-U\vert_{\{Z=a_i\}}.$$
Glue over $i$ to obtain random variables $U'$ and $V'$. Clearly, $U'\sim U$, $V' \sim Z-U\leq V$, and $X = U'+V'$.
Hence, $f(X)=f(U')+f(V')\leq \phi(U)+\phi(V)$ by Lemma \ref{discretization}. This establishes subadditivity of $\phi$.

Finally, consider $X\in \cX$, $m\in \R$ and  $Z\sim X+m\mathbf{1}$. Then $Z-m\mathbf{1}\sim X$. Since $f$ vanishes on $L^\infty$,  
$f(Z)  = f(Z-m\one)  \leq {\phi}(X).$
Taking supremum over $Z$ gives $\phi(X+m\mathbf{1})\leq \phi(X)$. The same inequality gives
\[{\phi}(X)={\phi}(X+m\mathbf{1}+(-m)\mathbf{1})\leq {\phi}(X+m\mathbf{1}).\] Therefore, ${\phi}(X+m\mathbf{1})={\phi}(X)$.
\end{proof}

With Lemma~\ref{p3.4} at hand, we now present the proof of  (1)$\implies$(3) in Theorem \ref{main-thm}.

\begin{proof}[Proof of (1)$\implies$(3) in Theorem \ref{main-thm}]
Suppose that $\cX$ fails the AOCEA property.
Take $X_0\in \cX_+$ such that $ \mathrm{d}(\mathcal{CL}(X_0),{\cX_a}) >0$, where $\mathcal{CL}(X_0) = \mathrm{co}\{X:X \sim X_0\}$.  We will show that condition (1) of Theorem \ref{main-thm} fails.

First, we claim that $\mathrm{d}\big(\mathcal{CL}(X_0), \mathrm{co}(\cX_a\cup (-\cX_+))\big) > 0$.
Indeed, since $\cX_a $ and $-\cX_+$ are both convex, a general element of $\mathrm{co}(\cX_a\cup (-\cX_+))$ is of the form $\alpha V- (1-\alpha)W$, where $0\leq \alpha \leq1$, $V\in \cX_a$ and $W \in \cX_+$. Using the notation  in Example \ref{counterexam}, for any $Z\in\mathcal{CL}(X_0)$,  $[Z + (1-\alpha)W]\geq [Z]\geq [0]$ in the quotient space $\cX/{\cX_a}$.  Hence
\begin{align*}
\bignorm{Z - \big(\alpha V- (1-\alpha)W\big)} &\geq \bignorm{[Z - (\alpha V- (1-\alpha)W)]}_q \\&= \bignorm{[Z + (1-\alpha)W]}_q \\&
\geq \norm{[Z]}_q = \mathrm{d}(Z,\cX_a).
\end{align*}
This proves that $\mathrm{d}\big(\mathcal{CL}(X_0), \mathrm{co}(\cX_a\cup (-\cX_+))\big) \geq \mathrm{d}(\mathcal{CL}(X_0),{\cX_a}) > 0$, as claimed.

Let $\mathcal{B}$ denote  the open unit ball of $\cX$.
By the claim, there exists $r>0$ such that $\mathcal{CL}(X_0)$ and $\mathrm{co}(\cX_a\cup (-\cX_+)) + r\mathcal{B}$ are disjoint (convex) sets in $\cX$.  Since the latter set is an open set,
the Hahn-Banach Separation Theorem (\cite[Theorem 3.4]{Rudin}) says that there is a nonzero linear functional $f\in \cX^*$ such that
\[ \sup\{f(X) : X \in \mathrm{co}(\cX_a\cup (-\cX_+))+r\mathcal{B}\} \leq  \inf\{f(Z) : Z\in \mathcal{CL}(X_0)\}.\]
In particular, $\sup\{f(X): X\in \cX_a\}<\infty$ and $\sup\{f(X): X\leq 0\} < \infty$. Consequently, since $\cX_a$ is a linear space and  $\{X\leq 0\}$ is a cone, $f =0$ on $\cX_a$ and $f(X)\leq 0$ if  $X\leq 0$. It follows from the latter conclusion that  $f$ is positive. 
Furthermore, since $f\neq 0$,
\begin{align}\label{2.9}
0=f(0)&\leq  \sup\{f(X) : X \in  \mathrm{co}(\cX_a\cup (-\cX_+))\} \nonumber
\\
&< \sup\{f(X) : X \in  \mathrm{co}(\cX_a\cup (-\cX_+))\}+r\norm{f}  \\&= \sup\{f(X) : X \in  \mathrm{co}(\cX_a\cup (-\cX_+))+r\mathcal{B}\} \nonumber\\&\nonumber \leq  \inf\{f(Z) : Z\in \mathcal{CL}(X_0)\} := \beta.\end{align}
Define  $ \phi:\cX\rightarrow \R$ by $\phi(X)=\sup\{f(X'): X'\sim X\}$ for $X\in\cX$. Set $$\rho(X) = \phi(-X)-\E[X], \text{for any }X\in\cX.$$
Invoking Lemma \ref{p3.4}, one sees that $\rho$ is a law-invariant coherent risk measure on $\cX$. 

It remains to show that $\rho$ fails the Fatou property at $0$.
For any $n\geq 1$, let $X_n = X_0\mathbf{1}_{\{X_0\geq n\}}$.
Then $(X_n)$ order converges to $0$. Let's compute $\rho(X_n)$.
Suppose that $Z\sim -X_n$.  Then $-Z\geq 0$ and $\mathbb{P}(-Z >0) = \mathbb{P}(X_0\geq n)$ so that  there exists a random variable $W$, supported on $\{-Z >0\}^c=\{Z=0\}$, such that $W \sim -X_0\mathbf{1}_{\{X_0<n\}}$.
Clearly,  $-(Z+W) \sim X_0$, and   by \eqref{2.9}, $f(-(Z+W)) \geq \beta > 0$.
Since $W\in L^\infty\subset \cX_a$, $f(W) =0$, implying that   $f(Z)  \leq -\beta < 0$. Thus $\phi(-X_n) \leq -\beta<0$, and consequently,   $ \rho(X_n)\leq -\beta-\mathbb{E}(X_n)$ for all $n\in\N$.
It follows that
\[ \liminf_n  {\rho}(X_n) \leq -\beta < 0 =  {\rho}(0).\]
Thus ${\rho}$ fails the Fatou property at $0$, i.e.,
condition (1) of Theorem \ref{main-thm} fails.
\end{proof}

\section{Automatic Representations}

In this section, building on Theorem \ref{main-thm} and the techniques developed for the proof, we obtain    automatic $\sigma(\cX,\cX')$ lower semicontinuity  and corresponding dual representations     of law-invariant risk measures.

\subsection{Automatic $\sigma(\cX,\cX')$ lower semicontinuity}\label{autocontinuity}

The following lemma is a refinement of Proposition \ref{MainLemma}(2) and hence its conclusion is another equivalent formulation of the AOCEA property. Recall that for a (measurable) partition $\pi = \{C_1, \dots, C_k\}$ of $\Omega$,  $$\E[X|\pi]:=\sum_{j=1}^k\frac{\E[X\one_{C_j}]}{\bP(C_j)}\one_{C_j},\quad\text{ for any }X\in L^1.$$

\begin{lemma}\label{improvedlemma}
Let $\cX$ be an r.i.\ space over a non-atomic probability space $(\Omega,\mathcal{F},\bP)$ other than $L^\infty$.
Suppose that $\mathcal{X}$ satisfies the AOCEA property.  Let $X\in \mathcal{X}$, a finite partition  $\pi = \{C_1, \dots, C_k\}$ of $\Omega$ and $\varepsilon>0$ be given. Then there exist  $X_1, \dots, X_N \in \mathcal{X}$ and a convex combination $X'=\sum_{i=1}^N\lambda_i X_i$ such that
$$ X_i \vert_{C_j} \sim X\vert_{C_j} \quad \text{for any } 1\leq i\leq N \text{ and } 1\leq j\leq k,$$
$$\mathrm{d}(X', \mathcal{X}_a)<\varepsilon.$$
In particular, $X_i\sim X$ for each $i=1,\dots,N$ and $\mathbb{E}[X'|\pi]=\mathbb{E}[X|\pi]$.
\end{lemma}

\begin{proof}
We demonstrate the proof for $k=2$; the same argument  applies for other values of $k$. Write $\pi=\{A,B\}$. Let $X\in\cX$ and $\varepsilon>0$ be given. By Proposition \ref{MainLemma}(2),   there exist random variables $(X_i')_{i=1}^m$, a convex combination $\sum_{i=1}^m\alpha_iX_i'$, and a random variable $V'\in\cX_a$  such that all $X_i'$'s are supported in $A$, $X_i'\sim X\mathbf{1}_A$ for each $i=1,\dots,m$, and $$\Bignorm{\sum_{i=1}^m\alpha_iX_i'-V'}<\frac{\varepsilon}{2}.$$
Similarly, there exist random variables $(X_j'')_{j=1}^l$, a convex combination $\sum_{j=1}^l\beta_jX_j''$, and a random variable $V''\in\cX_a$  such that  all $X_j''$'s are supported in $B$, $X_j''\sim X\mathbf{1}_B$ for each $j=1,\dots,l$, and   $$\Bignorm{\sum_{j=1}^l\beta_jX_j''-V''}<\frac{\varepsilon}{2}.$$
For any $i=1,\dots,m$ and $j=1,\dots,l$, put $X_{ij}=X_i'+X_j''$. Then $$X_{ij}\vert_{A}=X_i'\vert_{A}\sim X\vert_{A}\;\text{ and }\;X_{ij}\vert_{B} =X_j''\vert_{B}\sim X\vert_{B}.$$ This  implies in particular  that $X_{ij}\sim X$  and  $\E[X_{ij}\vert\pi]=\E[X\vert\pi]$.
The convex combination $X':=\sum_{1\leq i\leq m, 1\leq j\leq l}\alpha_i\beta_j X_{ij}=\sum_{i=1}^m\alpha_iX_i'+\sum_{j=1}^l\beta_jX_j''$  and the random variable $V:=V'+V''\in\cX_a$ clearly satisfy $\E[X'\vert\pi]=\E[X\vert \pi]$ and $$\norm{X'-V}<\varepsilon.$$
This completes the proof by reordering $X_{ij}$'s into $(X_i)_{i=1}^N$.
\end{proof}

The following lemma was proved for $L^\infty$ in \cite{Joui06} and for Orlicz spaces in \cite{GaoLeungMunariXanthos2017}. But new techniques are needed to extend it to general r.i.\ spaces.
\begin{lemma}\label{previous}
Let $\cX$ be an r.i.\ space over a non-atomic probability space $(\Omega,\mathcal{F},\bP)$ other than $L^\infty$.
Suppose that $\mathcal{X}$ satisfies the AOCEA property.   Let $\mathcal{C}$ be a norm-closed, convex, law-invariant set in $\cX$. Let $\pi = \{C_1, \cdots, C_k\}$ be a finite partition of $\Omega$.
Then $\mathbb{E}[X\vert \pi] \in \mathcal{C}$ for any $X\in \mathcal{C}$.
\end{lemma}

\begin{proof}
Let $X\in \mathcal{C}$ and $\varepsilon >0$ be given. We demonstrate the proof for $k=2$; the other cases can be proved similarly. Write $\pi = \{A, B\}$. We obtain the convex combination $X'=\sum_{i=1}^N\lambda_iX_i$ as in Lemma \ref{improvedlemma}. Then $$\mathbb{E}[X'\vert \pi] = \mathbb{E}[X\vert \pi]\quad\text{ and }\quad\norm{X'-V}<\varepsilon$$ for some $V\in \mathcal{X}_a$.  For each $i$, since $X_i\sim X$, $X_i\in\mathcal{C}$ by law-invariance of $\mathcal{C}$. Thus $X'\in\mathcal{C}$ by convexity of $\mathcal{C}$.
Recall that simple functions are norm dense in $\mathcal{X}_a$. Thus we may assume that $V$ is simple. Since $\lim_{\mathbb{P}(C)\rightarrow 0} \norm{\mathbf{1}_C}=0$, we may further assume, by perturbation if necessary, that
$$V\mathbf{1}_A = \sum_{j=1}^{m_1} a_j \mathbf{1}_{A_j},  \text{ where } \mathbb{P}(A_j) = \frac{\mathbb{P}(A)}{m_1}\text{ for }1\leq j\leq m_1,$$
$$V\mathbf{1}_B = \sum_{l=1}^{m_2} b_l \mathbf{1}_{B_l},  \text{ where } \mathbb{P}(B_l) = \frac{\mathbb{P}(B)}{m_2}\text{ for }1\leq l\leq m_2,$$
where $A_j$'s and $B_l$'s form a partition of $A$ and $B$, respectively.
For any permutation $\tau$ on $ \{1, \cdots,m_1\}$ and any permutation $\sigma$ on $  \{ 1, \cdots, m_2\}$, let $V_{(\tau,\sigma)}$ be the random variable defined by
$$V_{(\tau, \sigma)}\vert_{A_j} = a_{\tau(j)}, \quad j = 1, \cdots, m_1,$$
$$V_{(\tau, \sigma)}\vert_{B_l} = b_{\sigma(l)}, \quad l = 1, \cdots, m_2,$$
Using non-atomicity and the fact that $\bP(A_j)=\bP(A_{\tau(j)})$ and $\bP(B_l)=\bP(B_{\sigma(l)})$, we can also find a random variable $X_{(\tau,\sigma)}$  such that
$$X_{(\tau, \sigma)}\vert_{A_j} \sim X'\vert_{A_{\tau(j)}}, \quad j = 1, \cdots, m_1$$
$$X_{(\tau, \sigma)}\vert_{B_l} \sim X'\vert_{B_{\sigma(l)}}, \quad l = 1, \cdots, m_2.$$
Clearly, $X_{(\tau,\sigma)}\sim X'$ so that $X_{(\tau,\sigma)}\in\mathcal{C}$. By convexity of $\mathcal{C}$, $$\frac{1}{m_1!m_2!}\sum_{\tau, \sigma} X_{(\tau, \sigma)} \in \mathcal{C}.$$

Moreover, $X_{(\tau, \sigma)}\vert_{A_j} - V_{(\tau, \sigma)}\vert_{A_j} = X_{(\tau, \sigma)}\vert_{A_j}-a_{\tau(j)} \sim X'\vert_{A_{\tau(j)}}-a_{\tau(j)} = X'\vert_{A_{\tau(j)}}- V\vert_{A_{\tau(j)}}$.
Similarly,  $X_{(\tau, \sigma)}\vert_{B_l} - V_{(\tau, \sigma)}\vert_{B_l} \sim X'\vert_{B_{\sigma(l)}}- V\vert_{B_{\sigma(l)}}$. Thus
\begin{align}\label{swaptogether}
X_{(\tau, \sigma)}-V_{(\tau, \sigma)}\sim X'-V.\end{align}
In particular, $\norm{X_{(\tau, \sigma)}-V_{(\tau, \sigma)}} = \norm{X'-V}<\varepsilon$, so that
\begin{align*}
 \Bignorm{\frac{1}{m_1!m_2!}\sum_{\tau, \sigma} X_{(\tau, \sigma)}-\frac{1}{m_1!m_2!}\sum_{\tau, \sigma} V_{(\tau, \sigma)}}
 \leq  \frac{1}{m_1!m_2!}\sum_{\tau, \sigma} \norm{X_{(\tau, \sigma)}-V_{(\tau, \sigma)}}<\varepsilon.  \end{align*}
Furthermore,
\begin{align*}
    \frac{1}{m_1!m_2!}\sum_{\tau, \sigma} V_{(\tau, \sigma)} = \mathbb{E}[V\vert \pi]
\end{align*}
and
\begin{align*}
\bignorm{\mathbb{E}[V\vert \pi]-\mathbb{E}[X'\vert \pi]}
=& \Bignorm{\frac{1}{\bP(A)}\E[(V-X')\mathbf{1}_A]\mathbf{1}_A+ \frac{1}{\bP(B)}\E[(V-X')\mathbf{1}_B]\mathbf{1}_B}  \\
\leq &\Big(\frac{1}{\mathbb{P}(A)} + \frac{1}{\mathbb{P}(B)}\Big) \norm{V-X'}_{L^1}\norm{\mathbf{1}} \leq C \norm{V-X'} <C\varepsilon,
\end{align*}
where $C$ is a constant depending only on $\cX$ and $\pi$; cf.\ \eqref{l1norm}.
Hence, in view of $\mathbb{E}[X'\vert \pi] = \mathbb{E}[X\vert \pi]$, we have
\begin{align*}
    &\,\Bignorm{\frac{1}{m_1!m_2!}\sum_{\tau, \sigma} X_{(\tau, \sigma)}-\mathbb{E}[X\vert \pi]} \\
    \leq &\,\Bignorm{\frac{1}{m_1!m_2!}\sum_{\tau, \sigma} X_{(\tau, \sigma)}-\frac{1}{m_1!m_2!}\sum_{\tau, \sigma} V_{(\tau, \sigma)}} +\Bignorm{\frac{1}{m_1!m_2!}\sum_{\tau, \sigma} V_{(\tau, \sigma)}-\mathbb{E}[V\vert \pi]}\\
    &\quad+ \norm{\mathbb{E}[V\vert \pi]-\mathbb{E}[X'\vert \pi]} \\ \leq&\, (C+1)\varepsilon.
\end{align*}
Since $\frac{1}{m_1!m_2!}\sum_{\tau, \sigma} X_{(\tau, \sigma)} \in \mathcal{C}$ and $\mathcal{C} $ is norm closed, we get $\mathbb{E}[X\vert \pi]\in \mathcal{C}$.
\end{proof}

The critical idea in the proof is that while we easily swap $V$ around to average to $\E[V|\pi]$,  we need to swap $X$ in a way that  the key property \eqref{swaptogether} is maintained.

The proof of the theorem below is standard, once  one is furnished with Lemma \ref{previous}. We include a proof for the sake of completeness.
\begin{theorem}\label{autocont}
Let $\cX$ be an r.i.\ space over a non-atomic probability space $(\Omega,\mathcal{F},\bP)$ other than $L^\infty$.
Suppose that $\mathcal{X}$ satisfies the AOCEA property.   Let $\rho: \cX\rightarrow \mathbb{R}$ be convex,  decreasing, and law invariant. Then $\rho$ is $\sigma(\cX, \cX')$ lower semicontinuous  at every $X\in \cX$ such that $X^- \in \cX_a$.
\end{theorem}

\begin{proof}
Let $X\in \cX$ be such that $X^- \in \cX_a$. By \cite[Lemma 2]{FR20},  there is a sequence $(\pi_k)$ of finite partitions of $\Omega$  such that
$\mathbb{E}[X\vert \pi_k]\xrightarrow{o} X$ in $\cX$.
Since $\rho$ has the  Fatou property at $X$ by Theorem \ref{main-thm}, $$ \rho(X)\leq \liminf_k \rho(\mathbb{E}[X\vert \pi_k]).$$

Let $X_\alpha\xrightarrow{\sigma(\cX, \cX')}X$. For any $A\in\mathcal{F}$, since $\one_A\in L^\infty\subset \cX'$, $\E[X_\alpha\one_A]\rightarrow\E[X\one_A]$. Thus by the definition of $\E[\cdot\vert \pi]$, one sees that for any $ k\in\mathbb{N}$,  $\mathbb{E}[X_\alpha\vert \pi_k]\rightarrow \mathbb{E}[X\vert \pi_k]$ in the norm  topology of $\cX$. Recall from \cite[Proposition 3.1]{RS06} that $\rho$ is norm continuous. It follows that
$$\rho(\mathbb{E}[X|\pi_k]) = \lim_\alpha \rho(\mathbb{E}[X_\alpha\vert \pi_k]).$$
For each $\alpha$, the set $\{\rho\leq \rho(X_\alpha)\}$ is norm closed, by norm continuity of $\rho$ again. Since it is also convex and law invariant and contains $X_\alpha$,  Lemma \ref{previous}  implies that $\mathbb{E}[X_\alpha\vert \pi_k]\in \{\rho\leq \rho(X_\alpha)\}$, i.e., $\rho(\mathbb{E}[X_\alpha\vert \pi_k])\leq \rho (X_\alpha)$ for any $k\geq 1$. Therefore,
$$\rho(\mathbb{E}[X|\pi_k]) = \lim_\alpha \rho(\mathbb{E}[X_\alpha\vert \pi_k])\leq \liminf_\alpha \rho(X_\alpha).$$
Taking $\liminf$ over  $k$, we obtain $\rho(X)\leq \liminf_\alpha \rho(X_\alpha) $.
\end{proof}

\subsection{Automatic Dual Representations}
The well-known Fenchel-Moreau Duality asserts that if $\rho:\cX\rightarrow(-\infty,\infty]$ is proper (i.e., not identically $\infty$), convex, and $\sigma(\cX,\cX')$ lower semicontinuous everywhere, then $\rho$ has a dual representation via the dual space $\cX'$ at every $X\in\cX$. In our framework, $\rho$, however,  only has $\sigma(\cX,\cX')$ lower semicontinuity locally, not everywhere, and as a result, the Fenchel-Moreau Duality cannot be  applied directly. Fortunately, the classical proof in \cite[Section 1.4]{B} can be  modified to recover the dual representation theorem locally. We include the complete proof here for the convenience of the reader.

Let $\cY$ be a locally convex topological vector space  and let $\cY^*$ be its continuous dual.
Let $\rho: \cY\to (-\infty,\infty]$ be a proper, convex functional.  Define the conjugate functional $\rho^*:\cY^*\to (-\infty,\infty]$ by
\[ \rho^*(F) = \sup_{Y\in \cY}\big(F(Y) - \rho(Y)\big),\quad F\in\cY^*.\]

\begin{lemma}\label{rep-lemma1}
Suppose that $\rho:\cY\to \R$ is convex and (topologically) lower semicontinuous at $Y_0\in \cY$.  For any real number $\lambda_0 < \rho(Y_0)$, there exist $F\in \cY^*$ and a real number $k >0$ such that
\[ F(Y_0) + k\lambda_0 < \inf\big\{F(Y)+k\rho(Y): Y \in \cY\big\}.\]
\end{lemma}

\begin{proof}
Choose $\lambda\in\R$ such that $\lambda_0 < \lambda < \rho(Y_0)$.
Since $\rho$ is lower semicontinuous at $Y_0$ and $\cY$ is locally convex, there exists a convex open neighborhood $\mathcal{O}$ of $Y_0$ such that $\mathcal{O} \subseteq \{\rho > \lambda\}$.
Let $\mathcal{A} = \mathcal{O} \times (-\infty,\lambda)$. Then $\mathcal{A}$ is an open  convex set in $\cY\times \R$. It is disjoint with the convex set $\mathcal{C}_\rho:=\{(Y,\mu)\in\cY\times\R: \rho(Y)\leq \mu\}$.
By the Hahn-Banach Separation Theorem (\cite[Theorem 3.4]{Rudin}), there exist a nonzero linear functional $(F,k)\in (\cY\times \R)^*=\cY^*\times\R$ and $\alpha\in\R$ such that
\begin{align}\nonumber &\sup\{F(Y) + k\mu: Y\in \mathcal{O}, \mu\in\R, \mu< \lambda\} \\
\label{e4}\leq &\alpha \leq \inf\{F(Y) + k\mu: Y\in\cY,   \mu \in\R, \mu \geq\rho(Y) \}.\end{align}

Fix $Y\in \mathcal{O}$.  By the first inequality in \eqref{e4},  $F(Y) + k\mu \leq \alpha$ for all $\mu < \lambda$. Hence $k \geq 0$.
If $k=0$, then $F\neq 0$.
Since $\rho$ is real-valued, $(Y,\rho(Y)) \in \mathcal{C}_\rho$ for all $Y \in \cX$. Thus the second inequality in \eqref{e4} implies  that $\inf\{F(Y): Y\in \cY\} \geq \alpha$. This is impossible, since $F$ is linear and nonzero.  Thus $k >0$.
Choose sufficiently small $\varepsilon >0$ such that $\lambda_0 +\varepsilon < \lambda$.
Then
$( Y_0 ,\lambda_0+\varepsilon) \in \mathcal{O}\times (-\infty,\lambda)$.
Hence
\begin{align*}
&F(Y_0) +k\lambda_0 < F(Y_0) + k(\lambda_0 + \varepsilon) \leq \alpha\\
\leq& \inf\{F(Y) + k\mu: Y\in\cY,   \mu \in\R, \mu \geq\rho(Y)\}=\inf\{F(Y)+k\rho(Y): Y\in \cY\}.
\end{align*}
\end{proof}

\begin{lemma} \label{rep-lemma2}
If $\rho:\cY\to \R$ is convex and lower semicontinuous at some $Y_0\in \cY$, then $\rho^*$ is proper.
\end{lemma}

\begin{proof}(Modified from \cite{B}).
Choose $\lambda_0\in \R$ such that $\lambda_0 < \rho(Y_0)$.
By  Lemma \ref{rep-lemma1},
there exist  $F\in \cY^*$ and  $k >0$ such that
\[m:= \inf\{F(Y) + k\rho(Y): Y \in \cY\} \in\R.\]
Hence
\[ \frac{-F}{k}(Y) -\rho(Y) \leq \frac{-m}{k} \text{ for all $Y \in \cY$}.\]
By the definition of $\rho^*$, it follows that  $\rho^*(\frac{-F}{k}) < \infty$.
\end{proof}

Now, let $\cY$ be a vector space and $\cY^\#$ be a vector space of linear functionals on $\cY$ separating points of $\cY$. Then $\cY$ with the topology $\sigma(\cY,\cY^\#)$ is a locally convex topological vector space, and its continuous dual $\cY^*$ is just $\cY^\#$.
Assume that $\rho:\cY\to \R$ is  convex and $\sigma(\cY,\cY^\#)$ lower semicontinuous at some $Y_0\in \cY$.  By Lemma \ref{rep-lemma2}, $\rho^*$ is proper.  We can thus define $\rho^{**}:\cY\to (-\infty,\infty]$ by
\[ \rho^{**}(Y) = \sup_{F\in \cY^\#}\big(F(Y) -\rho^*(F)\big).
\]

\begin{proposition}\label{general-rep}
If $\rho:\cY\to \R$ is  convex and $\sigma(\cY,\cY^\#)$ lower semicontinuous  at $Y_0\in\cY$, then $\rho^{**}(Y_0) = \rho(Y_0)$.
\end{proposition}

\begin{proof} (Modified from \cite{B}).
By the definitions of $\rho^*$ and $\rho^{**}$, it is clear that $\rho^{**}(Y_0) \leq \rho(Y_0)$.
Assume by way of contradiction that $\rho^{**}(Y_0) <\rho(Y_0)$.
By  Lemma \ref{rep-lemma1}, there are $F\in \cY^\# $ and $k>0$ such that
\[ F(Y_0) + k \rho^{**}(Y_0) < \inf\{F(Y) + k\rho(Y): Y \in \cY\}.\]
Then
\begin{align}\label{firstin}
     \frac{-F}{k}(Y_0) - \rho^{**}(Y_0) > \sup\left\{\frac{-F}{k} (Y ) -\rho(Y):Y \in \cY\right\} = \rho^*\Big(\frac{-F}{k}\Big).
\end{align}
In particular, $\rho^*\Big(\frac{-F}{k}\Big) < \infty$.
But then by \eqref{firstin},
\[ \rho^{**}(Y_0) < \frac{-F}{k}(Y_0) - \rho^*\Big(\frac{-F}{k}\Big),
\]
contradicting the definition of $\rho^{**}(Y_0)$.
\end{proof}

Taking $\cY=\cX$ and $\cY^\#=\cX'$ in Proposition \ref{general-rep} and applying Theorem \ref{autocont}, we obtain the following automatic dual representation theorem.
\begin{theorem}\label{dualrep}
Let $\cX$ be an r.i.\ space over a non-atomic probability space $(\Omega,\mathcal{F},\bP)$ other than $L^\infty$.
Suppose that $\mathcal{X}$ satisfies the AOCEA property.  Let $\rho: \cX\rightarrow \mathbb{R}$ be convex, decreasing, and law invariant. Then  for any $X\in \cX$ with $X^- \in \cX_a$,
\begin{align}\label{rep-formula}
    \rho(X)=\sup_{Y\in\cX'}\big(\E[XY]-\rho^*(Y)\big),
\end{align}where $\rho^*(Y)=\sup_{X\in\cX}(\E[XY]-\rho(X))$ for any $Y\in\cX'$.
\end{theorem}

We end the paper with two more equivalent formulations of the AOCEA property.

\begin{remark}
It is well known that if $\rho$ has the dual representation \eqref{rep-formula} at some $X\in\cX$, then it is $\sigma(\cX,\cX') $ lower semicontinuous at $X$. It is also well known that if $\rho$ is $\sigma(\cX,\cX') $ lower semicontinuous at some $X\in\cX$, then it has the Fatou property at $X$. Thus we can add two further equivalent statements in Theorem \ref{main-thm}:
\begin{enumerate}
    \item[(1'')] Every convex, decreasing, law-invariant functional  $\rho:\cX\to \R$ is  $\sigma(\cX,\cX') $ lower semicontinuous at every $X\in\cX$  such that  $X^-\in \cX_a$.
     \item[(1''')] Every convex, decreasing, law-invariant functional  $\rho:\cX\to \R$ has the representation \eqref{rep-formula} at every $X\in\cX$  such that  $X^-\in \cX_a$.
\end{enumerate}
\end{remark}

\appendix
\section{Proof of Proposition~\ref{MainLemma}}\label{proofAOCEA}

\begin{lemma}\label{movingorder}
Let $X_1,X_2$ be random variables  on a non-atomic probability space $(\Omega,\mathcal{F},\mathbb{P})$ and $X_1'$ be a random variable on a non-atomic probability space $(\Omega',\mathcal{F}',\mathbb{P}')$. If $X_1'\sim X_1\geq X_2$, then there exists a random variable $X_2'$ on $\Omega'$ such that $$X_1'\geq X_2'\sim X_2.$$
The conclusion still holds if both ``$\geq $'' are replaced by ``$\leq$''.
\end{lemma}

\begin{proof}
Suppose that $X_1'\sim X_1\geq X_2$. 
Let $F_i$ and $q_i$ be the CDF and quantile of $X_i$, respectively. Then $q_1$ is also a quantile function of $X_1'$.
Since $X_1 \geq X_2$, $F_1\leq F_2$,  and thus $q_1\geq q_2$. 
Recall that there exists a random variable with uniform distribution on  $(0,1)$ such that $X_1'=q_1(U)$ (\cite[Lemma A.32]{HS4}). Thus
$$X_1'=q_1(U)\geq q_2(U)\sim X_2,$$
by \cite[Lemma A.23]{HS4}. Thus  it is enough to let $X_2'=q_2(U)$.
\end{proof}

\begin{proof}[Proof of Proposition~\ref{MainLemma}]
Assume that (4) holds. We show that (2) holds. Let   $A\in\mathcal{F}$ with $\bP(A)>0$,  $X\in \cX$ supported in $A$,  and   $\varepsilon>0$ be given. Put $A_n=\{\abs{X}\geq n\}$ for $n\in\N$. Then $A_n\downarrow\emptyset$.  By the assumption (4), there exist natural numbers $(n_i)_{i=1}^k$   and random variables  $(Z_{ i})_{i=1}^k$ such that $Z_i\sim X \one_{A_{n_i}}$ for $i=1,\dots,k$, all $Z_i$'s are supported in $A$. Moreover, a convex combination $\sum_{i=1}^k \lambda_iZ_{i}$ satisfies  $$\Bignorm{\sum_{i=1}^k \lambda_iZ_{i}}<\varepsilon.$$
Since $\{X \one_{A_{n_i}}\neq 0\}=\{ \abs{X}\geq {n_i}\}=A_{n_i}$,
$$Z_{ i}\vert_{\{Z_{ i} \neq 0\}}\sim X\vert_{A_{n_i}}.$$ Since $Z_i$ is supported in $A$, $\{Z_i\neq 0\}\subset A$; since $X$ is supported in $A$, $A_{n_i}\subset A$. Thus it is  immediate to see that   $\bP(A\backslash \{Z_i\neq 0\})=\bP(A\backslash A_{n_i})$. By non-atomicity, we can find a random variable $W_i$ such that $$W_i\vert_{A\backslash\{Z_{ i}\neq0\}}\sim X\vert_{A\backslash A_{n_i}}\text{ and }W_i\text{ vanishes off }A\backslash\{Z_{ i}\neq0\}.$$
For $i=1,\dots,k$, set $X_i=Z_i+W_i$. Clearly, $X_i\sim X$ for all $i$ and $X_i$'s are supported in $A$. Moreover, since $X$ is bounded by $n_i$ on $A_{n_i}^c$, $W_{ i}$ is bounded by $n_i$. In particular, $W_i\in L^\infty\subset \cX_a$. Thus $V:=\sum_{i=1}^k\lambda_iW_i\in\cX_a$.  Clearly, $V$ is  supported in $A$ and $\norm{\sum_{i=1}^k\lambda_iX_i-V}=\norm{\sum_{i=1}^k\lambda_iZ_i}<\varepsilon$.  This proves that (4)$\implies$(2).

Taking $A=\Omega$, the same argument gives (3)$\implies$(1). Since (2)$\implies$(1) and (4)$\implies$(3) are obvious, we have (4)$\implies$(2)$\implies$(1) and (4)$\implies$(3)$\implies$(1). To complete the proof, we show that (1)$\implies$(3)$\implies$(4).

Assume that (1) holds. We  show that (3) holds. Let $X \in \cX_+ $, $A_n\downarrow \emptyset$ and  $\varepsilon >0$ be given.
By the assumption (1), there is a convex combination $\sum^{m_1}_{i=1}\alpha_iX_i$ such that $X_i\sim X$ for $1\leq i\leq m_1$ and $\norm{\sum^{m_1}_{i=1}\alpha_iX_i - V} < 1$ for some  $V\in\cX_a$. In view of $\abs{a^+-b^+}\leq \abs{a-b}$ and $\sum^{m_1}_{i=1}\alpha_iX_i\geq 0$, we may replace $V$ with $V^+$ so that $V\geq 0$.
For any $i=1,\dots,m_1$, by $X_i\sim X\geq X\one_{A_i} $ and Lemma \ref{movingorder}, there exists a random variable  $Z_{i}$ such that  $$ X_i\geq  Z_i\sim X\mathbf{1}_{A_{i}}.$$
Set $$U_1:=V \wedge\Big(\sum_{i=1}^{m_1}\alpha_iZ_i\Big).$$
Clearly, $0\leq U_1\leq V$ so that  $U_1\in\cX_a $. In view of $a-a\wedge b= (a-b)^+ $, since $ {\sum_{i=1}^{m_1}\alpha_i Z_i}\leq  \sum_{i=1}^{m_1}\alpha_iX_i$, we have  $$\Bignorm{\sum^{m_1}_{i=1}\alpha_iZ_{i}-U_1} =\Bignorm{\Big(\sum^{m_1}_{i=1}\alpha_iZ_{i}-V\Big)^+}\leq \Bignorm{\Big(\sum^{m_1}_{i=1}\alpha_iX_{i}-V\Big)^+} <1.$$

Applying the same arguments to  the sequence $\{A_{m_1+n}\}_{n=1}^\infty$, we obtain $m_2>m_1$, $(Z_i)_{i=m_1+1}^{m_2}$, and $U_2\in\cX_a$ such that $Z_i\sim X\mathbf{1}_{A_i}$ for $i=m_1+1,\dots, m_2$ and $
\Bignorm{\sum_{i=m_1+1}^{m_2}\alpha_iZ_i-U_2}<\frac{1}{2}$.
Repeating this process, we get $(Z_i)_{i=1}^\infty$, convex combinations $(\sum_{i=m_{j-1}+1}^{m_{j}}\alpha_iZ_i)$, and $(U_j)\subset \cX_a$ such that
\begin{align*}
 Z_i\sim X\mathbf{1}_{A_i}\quad\text{ for each }i\in\N,
\end{align*}
\begin{align}\label{Cons2}
  \Bignorm{ \sum_{i=m_{j-1}+1}^{m_{j}}\alpha_iZ_i-U_j }<\frac{1}{j}\quad\text{for each }j\in\N.
\end{align}

For a random variable $W$, let $W^*$ be its decreasing rearrangement given by $W^*(t)=\inf\big\{\lambda>0 : \bP(\abs{W}>\lambda)\leq t\big\}$, $t\in(0,1)$.
Let $W\in \cX'$.  By the Hardy-Littlewood Inequality (\cite[Chp 2, Theorem 2.2]{BS88}),
$$\abs{\mathbb{E}[WZ_n]}\leq \int_0^1 W^*(Z_n)^*\,\mathrm{d}t\leq \int_0^1 W^*X^* \mathbf{1}_{[0,\mathbb{P}(A_n)]}\,\mathrm{d}t\rightarrow0, $$
since $W^*X^*\in L^1$ (\cite[Chp 2, Theorem 2.6]{BS88}) and $\mathbb{P}(A_n)\rightarrow0$.   Thus  as $j\rightarrow\infty$,
\begin{align}\label{2.4}
\mathbb{E}\Big[W \sum_{i=m_{j-1}+1}^{m_{j}}\alpha_iZ_i\Big]\rightarrow 0.
\end{align}
Since $W$ acts a bounded linear functional on $\cX$, it follows from \eqref{Cons2} and \eqref{2.4} that
$$\lim_j\mathbb{E}[WU_j]=0.$$
That is, $U_j\stackrel{\sigma(\cX,\cX')}{\longrightarrow}0$.
Recall from \cite[Lemma 3.3]{GLX} that $(\cX_a)^* = \cX'$. Thus $(U_i)$ converges to $0$ weakly in $\cX_a$.

Let $\varepsilon>0$. Take $j_0>\frac{2}{\varepsilon}$.
By Mazur's Theorem, $0\in \overline{\mathrm{co}(U_j)_{j\geq j_0}}^{\|\cdot\|}$. Thus there is a convex combination $\sum^{j_1}_{j=j_0}\beta_jU_j$ such that $$\Bignorm{\sum^{j_1}_{i=j_0}\beta_jU_j} <\frac{\varepsilon}{2}.$$
Let
\begin{align}\label{define}
    Z = \sum^{j_1}_{j=j_0}\beta_j\Big(\sum^{m_{j}}_{i=m_{j-1}+1}\alpha_i Z_i\Big).
\end{align}
Then $Z\in \mathrm{co}(Z_i)_{i=m_{j_0-1}+1}^{m_{j_1}}$ and
\begin{align*}
\norm{Z} \leq &\sum^{j_1}_{j=j_0}\beta_j\Bignorm{\sum^{m_{j}}_{i=m_{j-1}+1}\alpha_iZ_i - U_j} + \Bignorm{\sum^{j_1}_{j=j_0}\beta_jU_j}\\
\leq &\sum^{j_1}_{j=j_0}\beta_j  \frac{1}{j}+\frac{\varepsilon}{2}
<\sum^{j_1}_{j=j_0}\beta_j\frac{\varepsilon }{2}+\frac{\varepsilon}{2} = \varepsilon.
\end{align*}
Finally, rewrite $Z=\sum_{i=1}^k \lambda_iZ_i$, where $k=m_{j_1}$ and $\lambda_i=0$ if $Z_i$ is not involved in defining $Z$ in \eqref{define}. This proves that  (3) holds.

Now we show that (3)$\implies$(4). Let $X\in\cX$, $A\in\mathcal{F}$ with $\bP(A)>0$, $A_n\downarrow \emptyset$, and  $\varepsilon >0$ be given. By passing to a subsequence, we may assume without loss of generality that $\sum_{n=1}^\infty \bP(A_n)<\frac{\bP(A)}{2}$.
Divide $A$ as a disjoint union $B\cup C$, where $\bP(B) = \bP(C) = \frac{\bP(A)}{2}$.
Applying (3) to $\abs{X}$, we find $Z_i' \sim \abs{X}\one_{A_{n_i}}$, $i=1,\dots,k$, such that a convex combination   satisfies  $\norm{\sum_{i=1}^k\lambda_iZ_i'}<\frac{\varepsilon}{2}$.
Since $$\bP\big((Z_1',\dots,Z_k')\neq 0\big)=\bP\Big(\bigcup_{i=1}^k\{Z_i'\neq0\}\Big) \leq \sum_{i=1}^k\mathbb{P}(A_{i})     \leq \mathbb{P}(B),$$
by non-atomicity, we can find a random vector $(S_1,\dots,S_k) $ such that  $$(S_1,\dots,S_k)\text{ is supported in }B,\quad (S_1,\dots,S_k)\sim (Z_1',\dots,Z_k').$$
In particular, $\sum_{i=1}^k\lambda_iS_i\sim \sum_{i=1}^k\lambda_iZ_i'$  so that $\norm{\sum_{i=1}^k\lambda_iS_i}<\frac{\varepsilon}{2}$.
Since $S_i \sim \abs{X}\one_{A_{n_i}} \geq X^+\one_{A_{n_i}}$, by Lemma \ref{movingorder}, there exists a random variable $Q_i$ such that $$S_i\geq Q_i\sim X^+\one_{A_{n_i}}  .$$
Then $0\leq \sum_{i=1}^k\lambda_iQ_i\leq \sum_{i=1}^k\lambda_iS_i$, so that $\norm{\sum_{i=1}^k\lambda_iQ_i}<\frac{\varepsilon}{2}$. Clearly, all $Q_i$'s are supported in $B$.
Similarly, we obtain $R_i\sim X^-\one_{A_{n_i}}$, $i=1,\dots, k$, such that all $R_i$'s are supported in $C$ and $\norm{\sum_{i=1}^k\lambda_iR_i}<\frac{\varepsilon}{2}$.
Since $Q_i$ is supported in $B$ and $R_i$ is supported in $C$, $Q_i-R_i \sim X\one_{A_{n_i}}$ and $Q_i-R_i$ is supported in $A$. Finally,
\begin{align*} \Bignorm{\sum_{i=1}^k\lambda_i(Q_i-R_i)} \leq \Bignorm{\sum_{i=1}^k\lambda_iQ_i} + \Bignorm{\sum_{i=1}^k\lambda_iR_i}<\frac{\varepsilon}{2}+\frac{\varepsilon}{2}=\varepsilon.
\end{align*}
The proof of (3)$\implies$(4) is completed by setting $Z_i=Q_i-R_i$.
 \end{proof}

\section{An r.i.\ space failing the AOCEA property}\label{poorexample}
Endow $(0,1)$ with the Lebesgue measure $\bP$. For a random variable $X$ on $(0,1)$, let $X^*$ be the decreasing rearrangement of $X$ defined by 
\[ X^*(t) = \inf\{s > 0: \bP(|X| > s) \leq t\},\ t\in (0,1).
\]
We refer to \cite[Chapter 2]{BS88} for detailed properties of decreasing rearrangement.

\begin{example}\label{ex2.8}
Let
 $\cX$ be the space of all random variables  $X$ on $(0,1)$ such that 
\[ \|X\|: = \sup_{n\in\N}\,n2^n\int^{\frac{1}{2^n\cdot n!}}_0X^*\,\mathrm{d}t < \infty.\]
Then $\cX$  is an r.i.\ space over $(0,1)$  that fails the AOCEA property.
\end{example}

\begin{proof} Note that $\one^*=\one$. Thus it is easy to see that $\one\in \cX$ and hence $\cX\neq\{0\}$.
For each $n\in\N$, define $\tau_n$ on $\cX$ by $\tau_n(X) = n2^n\int^{\frac{1}{2^n\cdot n!}}_0X^*\,\mathrm{d}t$.   We have the inequality
\[ \int^s_0 (X+Y)^* \,\mathrm{d}t \leq \int^s_0 X^*\,\mathrm{d}t + \int^s_0 Y^*\,\mathrm{d}t\]
for any $X,Y\in L^1(0,1)$ and all $s\in (0,1)$; see, e.g., \cite[p.\ 54]{BS88}.
Thus each $\tau_n$, and hence $\|\cdot\|$, satisfies the triangle inequality.  It is then clear that $\tau_n$ is a seminorm on $\cX$.  
Moreover, as $X^*$ is a decreasing function and $X^* \sim |X|$, 
\[ \|X\|_1 = \|X^*\|_1 \leq 2^n\cdot n! \int^{\frac{1}{2^n\cdot n!}}_0X^*\,\mathrm{d}t \leq 2^n\cdot n!\,\|X^*\|_1 = 2^n\cdot n!\,\|X\|_1. 
\]
So each $\tau_n$ is in fact a lattice norm on $\cX$ that is equivalent to the $L^1$-norm.
In particular, $\|\cdot\|$ is a lattice norm on $\cX$.  Law-invariance of $\|\cdot\|$ is obvious.
To see that $\cX$ is an r.i.\ space, it suffices to show the completeness of $\|\cdot\|$.
Let $(X_k)_{k=1}^\infty$ be a norm Cauchy sequence in $\cX$.  Since each $\tau_n\leq \|\cdot\|$, $(X_k)$ is Cauchy in $\tau_n$-norm for all $n$.  By equivalence of $\tau_n$ with $L^1$-norm, 
 there exists $X\in L^1$ such that $(X_k)$ converges to $X$ with respect to $\tau_n$ for all $n$.
In particular, 
\[ \sup_n\tau_n(X) \leq \sup_n\sup_k\tau_n(X_k) \leq \sup_k\|X_k\| < \infty.\]
Thus $X\in \cX$.
Given $\varepsilon > 0$, choose $k_0\in\N$ so that $\|X_k-X_j\| \leq \varepsilon$ if $k,j\geq k_0$.  If $k \geq k_0$, then for any $n\in \N$, 
\[  \tau_n(X_k-X) = \lim_j\tau_n(X_k-X_j) \leq \limsup_j\|X_k-X_j\| \leq \varepsilon.
\]
Hence $\norm{X_k-X}\leq \varepsilon$ for any $k\geq k_0$. This shows that $(X_k)$ converges to $X$ in $\|\cdot\|$-norm and completes the proof that $\cX$ is an r.i.\  space.
 
Next, we show that $\cX\neq L^\infty$ and $\cX$ fails the AOCEA property. For convenience, set $c_n = \frac{1}{2^n\cdot (n+1)!}$ for all $n\in \N$.  Let $X$ be the function
\[ X =\sum^\infty_{n=1}n!\one_{[c_{n+1},c_n)}.\]
Then $X\geq 0$ and it is decreasing so that $X^*=X$ a.s. For any $m\geq 2$,
\begin{align*}
\int^{\frac{1}{2^m\cdot m!}}_0X^* \,\mathrm{d}t&\leq  \frac{(m-1)!}{2^m\cdot m!}+\sum^\infty_{n=m}n!c_n = \frac{1}{m2^m}+ \sum^\infty_{n=m}\frac{1}{2^n(n+1)} \leq \frac{3}{m2^m}.
\end{align*}
It follows that $X\in \cX$ and that $\cX\neq L^\infty$.
Now suppose that there exists $Y \in \mathrm{co}\{Z: Z\sim X\}$ and $U \in \cX_a$ so that $\|Y-U\| < \frac{1}{4}$.
There exists $m_0\in\N$ so that $\|U\one_A\| < \frac{1}{4}$ if $\bP(A) \leq \frac{1}{2^{m_0}\cdot m_0!}$.
Thus, if $m\geq m_0$, then
$m2^m\int^{\frac{1}{2^{m}\cdot m!}}_0U^* < \frac{1}{4}$.
Write $Y$ as a convex combination $\sum^m_{j=1}b_jZ_j$, where $Z_j\sim X$ for all $j$.
We may assume that $m \geq m_0\geq 2$.
Choose measurable sets $A_j$, $1\leq j\leq m$, so that $\bP(A_j) = \frac{1}{m2^m\cdot m!}$ and that 
$\E[Z_j\one_{A_j}] = \int^{ \frac{1}{m2^m\cdot m!}}_0Z^*_j \,\mathrm{d}t$. Then
\[ \E[Z_j\one_{A_j}] \geq  \int^{c_m}_0X^* \,\mathrm{d}t \geq m!\,c_m= \frac{1}{(m+1)2^m}.\]
Let $A = \bigcup^m_{j=1}A_j$.  Then $\bP(A) \leq\frac{1}{2^m\cdot m!}$.  Therefore,
\[ \int^{\frac{1}{2^m\cdot m!}}_0Y^* \geq \E[Y\one_A]  \geq \sum^m_{j=1}b_j\E[Z_j\one_{A_j}]  \geq \frac{1}{(m+1)2^{m}}.
\]
Thus
\begin{align*}
\frac{1}{2} &\leq m2^m\int^{\frac{1}{2^m\cdot m!}}_0Y^* \,\mathrm{d}t \leq m2^m\Big[\int^{\frac{1}{2^m\cdot m!}}_0(Y-U)^* \,\mathrm{d}t+ \int^{\frac{1}{2^m\cdot m!}}_0U^* \,\mathrm{d}t\Big]\\
&\leq \|Y-U\| + \frac{1}{4}.
\end{align*}
Hence $\|Y-U\| \geq \frac{1}{4}$, contrary to the choice of $Y$ and $U$.
\end{proof}

\noindent\textbf{Acknowledgement.} The authors would like to thank  Dr.\ Felix-Benedikt Liebrich for  simplifying the proof of Lemma \ref{movingorder} and Professor Alexander Schied  and Professor Ruodu Wang  for beneficial discussions. We also thank the editor and reviewers for their constructive feedback that improves the quality of the paper.


{\footnotesize

\begin{thebibliography}{99}
\bibitem{AA06}
Abramovich, Y.A., Aliprantis, C.D.: {\em An Invitation to Operator Theory}, Graduate Studies in Mathematics, 55, American Mathematical Society, Providence (2002)

	
\bibitem{AB}
Aliprantis, C.D., Burkinshaw, O.: {\em Positive Operators}, Springer, Dordrecht  (2006)

\bibitem{ArtznerDelbaenEberHeath1999} Artzner, Ph., Delbaen, F., Eber, J.-M., Heath, D.:
    Coherent measures of risk, {\em Mathematical Finance}, 9, 203-228 (1999)


\bibitem{BC}Bellini, F., Koch-Medina, P., Munari, C., Svindland, G.:
Law-invariant functionals on general spaces of random variables, {\em SIAM Journal on Financial Mathematics}, 12(1), 318-341 (2021)

\bibitem{BCb}Bellini, F., Koch-Medina, P., Munari, C., Svindland, G.:
Law-invariant functionals that collapse to the mean, {\em Insurance: Mathematics and Economics}, 98, 83-91 (2021)

\bibitem{BCc}
Bellini, F., Laeven, R.J.A., Rosazza Gianin, E.: Dynamic robust Orlicz premia and Haezendonck–Goovaerts risk measures, {\em European Journal of Operational Research}, 291 (2), 438-446 (2021)

\bibitem{BS88}
Bennet, C., Sharpley, R.: {\em Interpolation of Operators}, Academic Press, Inc., Boston (1988)

\bibitem{BF}Biagini, S., Frittelli, M.: On the extension of the Namioka-Klee theorem and on the Fatou Property for risk measures. In: Delbaen, F., Rasonyi, M., Stricker, C. (eds.) {\em Optimality and
Risk: Modern Trends in Mathematical Finance}, Springer, Berlin, 1-28  (2009)


\bibitem{B}
Brezis, H: {\em Functional Analysis, Sobolev Spaces and Partial Differential Equations}, Springer, New York (2011)

\bibitem{CGLLb} Chen, S., Gao, N., Leung, D., Li, L.: Do law-invariant linear functionals collapse to the mean? Preprint available at {\it  arXiv:2107.11239}.

\bibitem{CGX}
Chen, S., Gao, N., Xanthos, F.: On the strong Fatou property of risk measures, {\em Dependence Modeling}, 6(1), 183-196 (2018)

\bibitem{CheriditoLi2008} Cheridito, P., Li, T.: Dual characterization of properties of risk measures on Orlicz hearts, {\em Mathematics and Financial Economics} 2, 1-29 (2008)

\bibitem{CheriditoLi2009} Cheridito, P., Li, T.: Risk measures on Orlicz hearts, {\em Mathematical Finance}, 19, 189-214 (2009)



\bibitem{DK}
Drapeau, S., Kupper, M: Risk preferences and their robust representations,  {\em Mathematics of Operations Research}, 38, 28-62 (2013)

\bibitem{D} Delbaen, F.: Coherent risk measures on general probability spaces. In: Sandmann, K., Sch$\ddot{\text{o}}$nbucher,
P.J. (eds.) {\em Advances in finance and stochastics}, 1–37. Springer, Berlin (2002)


\bibitem{DO} Delbaen, F., Owari, K.: Convex functions on dual Orlicz spaces. {\em Positivity} 23(5), 1051–1064 (2019)

\bibitem{EdgarSucheston1992} Edgar, G.A., Sucheston, L.: {\em Stopping Times and Directed
    Processes}, Cambridge University Press, Cambridge (1992)

\bibitem{FarkasKochMunari2014} Farkas, W., Koch-Medina, P., Munari, C.: Beyond cash-additive risk measures: when changing the num\'{e}raire fails, {\em Finance and Stochastics}, 18, 145-173 (2014)

\bibitem{FS08}
Filipovi\'c, D., Svindland, G.: Optimal capital and risk allocations for law- and cash-invariant convex functions, {\em Finance and Stochastics}, 12(3), 423-439 (2008)
\bibitem{FS12} Filipovi\'c, D., Svindland, G.: The canonical model space for law-invariant convex risk measures is $L^1$,
{\em Mathematical Finance}, 22(3), 585-589 (2012)


\bibitem{HS4}F\"{o}llmer, H., Schied, A.: {\em Stochastic Finance An Introduction in Discrete Time}, 4th Edition, de Gruyter (2016)

\bibitem{FR}
Frittelli, M., Rosazza Gianin, E.: Putting order in risk measures, {\em Journal of Banking and Finance}, 26, 1473-1486 (2002)


\bibitem{GaoLeungMunariXanthos2017} Gao, N., Leung, D., Munari, C., Xanthos, F.: Fatou
    property, representations, and extensions of law-invariant risk measures on general
    Orlicz spaces, {\em Finance and Stochastics}, 22, 395-415 (2018)

\bibitem{GLX}
Gao, N., Leung, D., Xanthos, F.: Duality for unbounded order convergence and applications, {\em Positivity}, 22, 711-725 (2018)

\bibitem{GLX2}
Gao, N., Leung, D., Xanthos, F.: Closedness of convex sets in Orlicz spaces with applications to dual
representation of risk measures, {\em Studia Mathematica}, 249, 329-347 (2019)

\bibitem{GLL}
Gao, N., Liang, Y., Liebrich, F.B.:
Dilatation monotonicity and Schur convexity in general r.i.\ spaces, {\em preprint}
 
\bibitem{GMX}
Gao, N., Munari, C., Xanthos, F.: Stability properties of Haezendonck-Goovaerts
premium principles, {\em Insurance: Mathematics and Economics}, 94, 94-99 (2020)

\bibitem{GX}
Gao, N., Xanthos, F.: On the C-property and $w^*$-representations of risk measures, \emph{Mathematical Finance} 28(2), 748-754 (2018)

\bibitem{Joui06}
Jouini, E., Schachermayer, W., Touzi, N.: Law invariant risk measures have the Fatou Property, {\em Adv. Math. Econ.}, 9, 49-71 (2006)

\bibitem{Joui08}
Jouini E., Schachermayer W., Touzi N.: Optimal risk sharing for law invariant monetary utility functions, {\em Mathematical Finance}, 18(2), 269-292 (2008)

\bibitem{KSZ}
Kr\"{a}tschmer, V, Schied, A., Z\"{a}hle, H.: Comparative and qualitative robustness for law-invariant risk measures, {\em Finance and Stochastics}, 18, 271-295 (2014)

\bibitem{K01}
Kusuoka, S.: On law invariant coherent risk measures, {\em Advances in Mathematical Economics}, 3, 83-95 (2001)

\bibitem{LZ}
Lindenstrauss, J., Tzafriri, L.: {\em Classical Banach Spaces II}, Springer-Verlag (1979)

\bibitem{LCLW}
Liu, F., Cai, J., Lemieux, C., Wang, R.: Convex risk functionals: Representation and applications, {\em Insurance: Mathematics and Economics} 90, 66-79 (2020)
 

\bibitem{LW}
Liu, P., Wang, R., Wei, L.: Is the inf-convolution of law-invariant preferences law-invariant? {\em Insurance: Mathematics and Economics}, 91, 144-154 (2020)

\bibitem{MW}
Mao, T., Wang, R.: On aggregation sets and lower-convex sets, {\em Journal of Multivariate Analysis}, 138, 170-181 (2015)


\bibitem{MN}
Meyer-Nieberg, P.: {\em Banach Lattices}, Springer, Berlin  (1991)

\bibitem{Munari}
Munari, C.: {\em Measuring risk beyond the cash-additive paradigm}, PhD Dissertation, ETH Zurich (2015)


\bibitem{FR20}
Rahsepar, M., Xanthos, F.: On the extension property of dilatation monotone risk measures, {\em Statistics and Risk Modeling}, to appear. DOI: 10.1515/strm-2020-0006


\bibitem{Rudin}
Rudin, W.: {\em Functional Analysis}, McGraw-Hill, Inc., New York (1991)

\bibitem{RS06}Ruszczy\'{n}ski, A., Shapiro, A.: Optimization of Convex Risk Functions, {\em Mathematics of Operations Research}, 31(3), 433-452 (2006)

\bibitem{Svindland}
Svindland, G.: Continuity properties of law-invariant (quasi-)convex risk functions on $L^\infty$, {\em Mathematics and Financial Economics}, 3(1), 39-43  (2010)

\bibitem{TL}
Tantrawan, M., Leung, D.: On closedness of law-invariant convex sets in rearrangement invariant spaces, {\em Archiv der Mathematik}, 114(2), 175-183 (2020)

\bibitem{WR}
Wang, R., Zitikis, R.: An axiomatic foundation for the Expected Shortfall, {\em  Management Science}, 67(3), 1413-1429 (2021)

\bibitem{W}
Weber S.: Distribution-invariant risk measures, information, and dynamic consistency, {\em Mathematical Finance}, 16, 419-441 (2006)

\end{thebibliography}

}
\end{document}